\documentclass[11pt]{article}

\usepackage{fullpage} 
\usepackage{comment}
\usepackage{amsthm}
\usepackage{tikz}
\usetikzlibrary{positioning,calc,arrows.meta}
\usepackage{pgfplots}
\pgfplotsset{compat=1.18}
\usetikzlibrary{patterns}
\usepgfplotslibrary{fillbetween}
\usetikzlibrary{intersections}
\usepackage{pgfplots}

\usepackage{csquotes}

\usepackage[dvipsnames]{xcolor}

\usepackage{bigints}
\usepackage{amsmath,amssymb}
\usepackage{ifthen}   %
\usepackage{bbm}      %
\usepackage{xcolor}
\usepackage{tcolorbox}

\usepackage{thm-restate}

\usepackage{natbib}
\bibpunct{(}{)}{;}{a}{,}{,}

\usepackage[hidelinks]{hyperref}       %

\usepackage{comment}
\newcommand{\xhdr}[1]{\vspace{2mm}\noindent{\bf {#1}\ }}

\usepackage{url}            %
\usepackage{booktabs}       %
\usepackage{amsfonts}       %
\usepackage{nicefrac}       %
\usepackage{microtype}      %
\usepackage{dsfont}

\usepackage{mathtools, bm}

\usepackage{algorithm}
\usepackage{algpseudocode}

\usepackage{subcaption}
\usepackage{multirow}
\usepackage{float}
\usepackage{xspace}

\usepackage{enumitem}

\usepackage{xfrac}

\usepackage{cleveref}
\crefname{enumi}{part}{parts}

\theoremstyle{plain}
\newtheorem{theorem}{Theorem}[section]
\newtheorem{lemma}[theorem]{Lemma}
\newtheorem{claim}[theorem]{Claim}

\newtheorem{proposition}[theorem]{Proposition}
\newtheorem{corollary}[theorem]{Corollary}%

\newtheorem{informal}{Informal Theorem}

\theoremstyle{plain}
\newtheorem{definition}{Definition}[section] %

\allowdisplaybreaks

\usepackage{etoolbox}
\preto\part{\setcounter{section}{0}}

\newcommand{\dom}{\Omega}                        %
\newcommand{\metricspace}{\parent{\dom,\dst}}         %
\newcommand{\dst}{d}                             %
\newcommand{\norm}[1]{\left\lVert#1\right\rVert}  %
\newcommand{\dimension}{D}                       %
\newcommand{\agents}{N}                          %
\newcommand{\prof}{\mathbf{x}}                   %
\newcommand{\loc}[1]{x_{#1}}                     %
\newcommand{\pair}{\bm{\ell}}                    %
\newcommand{\pairs}[1]{\Delta\parent{#1^{2}}}    %

\newcommand{\cost}{\mathrm{cost}}                %
\DeclareMathOperator{\SC}{SC}                    %
\DeclareMathOperator{\OPT}{OPT}                  %
\newcommand{\ratio}{\rho}                        %
\newcommand{\Mech}{f}                            %
\newcommand{\Prop}{\mathsf{P}}                   %
\newcommand{\GPair}{\mathsf{G}}                  %
\newcommand{\Mix}[1]{\mathsf{M}_{#1}}            %
\newcommand{\pw}{W}                              %

\newcommand{\devd}{r}                            %
\newcommand{\tdist}[1]{u_{#1}}                   %
\newcommand{\rdist}[1]{u\primed_{#1}}                   %
\newcommand{\odist}[2]{d_{#1#2}}                 %
\newcommand{\Ssum}{S}                            %
\newcommand{\Tsum}{S\primed}                            %
\newcommand{\Aterm}{A}                           %
\newcommand{\Bterm}{B}                           %
\newcommand{\degr}[1]{S_{#1}}                    %

\newcommand{\lbprofA}{A}                         %
\newcommand{\lbprofB}{B}                         %
\newcommand{\Rlb}[2]{R\parent{#1,#2}}            %
\newcommand{\lbfun}{f}                           %
\newcommand{\lbconst}{\frac{1+\sqrt2}{2}}        %
\newcommand{\ofac}[1]{\ell^{*}_{#1}}               %
\newcommand{\clus}[1]{L_{#1}}                     %
\newcommand{\nclus}[1]{n_{#1}}                    %
\newcommand{\sep}{\delta}                         %
\newcommand{\ocost}[1]{c_{#1}}                   %
\newcommand{\Rsum}{C^{*}}                        %
\newcommand{\Qsum}{Q}                            %
\newcommand{\dispersion}{h}                     %

\newcommand{\kcon}{\kappa}                       %

\newcommand{\lamstar}{\lambda_{*}}               %
\newcommand{\rhostar}{\ratio_{*}}                %
\newcommand{\setsize}[1]{{\left|#1\right|}}
\newcommand{\set}[1]{\left\{#1\right\}}
\newcommand{\setfix}[1]{\{#1\}}
\newcommand{\parent}[1]{\left(#1\right)}
\newcommand{\parentfix}[1]{(#1)}

\newcommand{\lranglefix}[1]{\langle#1\rangle}

\newcommand{\deq}{\triangleq}

\newcommand{\reals}{\mathbb{R}}

\newcommand{\condition}{\,\mid\,}

\DeclareMathOperator*{\prb}{\mathrm{Pr}}
\newcommand{\prob}[2][]{\prb\ifthenelse{\not\equal{}{#1}}{\nolimits_{#1}}{}\!\left[{\def\givenn{\middle|}#2}\right]}
\newcommand{\expect}[2][]{\mathbb{E}\ifthenelse{\not\equal{}{#1}}{_{#1}}{}\!\left[{\def\givenn{\middle|}#2}\right]}

\newcommand{\tparen}{\big}
\newcommand{\tprob}[2][]{\text{Pr}\ifthenelse{\not\equal{}{#1}}{_{#1}}{}\tparen[{\def\given{\tparen|}#2}\tparen]}
\newcommand{\texpect}[2][]{\mathbb{E}\ifthenelse{\not\equal{}{#1}}{_{#1}}{}\tparen[{\def\given{\tparen|}#2}\tparen]}

\newcommand{\sprob}[2][]{\text{Pr}\ifthenelse{\not\equal{}{#1}}{_{#1}}{}[#2]}
\newcommand{\sexpect}[2][]{\mathbb{E}\ifthenelse{\not\equal{}{#1}}{_{#1}}{}[#2]}

\newcommand{\btheta}[1]{\Theta\parent{#1}}

\newcommand{\primed}{^{\prime}}

\title{Breaking the 4-Approximation Barrier in Strategyproof Two-Facility Location}

\author{
Mengfan Ma\thanks{Central China Normal University, China. Email: \url{mengfanma1@gmail.com}}
\and
Bo Peng\thanks{Shanghai University of Finance and Economics, China. Email: \url{ahqspbo@gmail.com}}
}

\date{ }

\begin{document}

\maketitle

\begin{abstract}
We study strategyproof mechanism design without transfers for the two-facility location problem in metric spaces. A mechanism selects two facility locations based on agents’ reported locations; each agent incurs her distance to the nearer facility, and the objective is to minimize the expected social cost. A mechanism is strategyproof if no agent ever benefits from misreporting her location. The best approximation ratio achieved by a randomized strategyproof mechanism has been $4$, attained by the Proportional mechanism of Lu, Sun, Wang, and Zhu (EC 2010), and the best lower bound has been $1.045$, due to Lu, Wang, and Zhou (WINE 2009). Neither bound has moved since then, even on the line $\mathbb{R}$.

We improve both bounds. Our main result is a randomized strategyproof mechanism with approximation ratio $11/3\approx3.667$ on every Ptolemaic metric space, a rich class containing all Euclidean spaces. The mechanism randomizes between the Proportional mechanism and a new mechanism that we call Global Pair. Global Pair draws an unordered pair of agents with probability proportional to their distance and opens facilities at their reported locations. Although Global Pair and Proportional each have approximation ratio $4$, the two mechanisms attain their worst-case approximation ratios on complementary instances. Randomizing between them balances these complementary weaknesses and breaks the $4$-approximation barrier. On the lower bound side, we construct a new two-profile instance that yields a lower bound of $(1+\sqrt{2})/2\approx 1.207$, improving upon the previous lower bound of $1.045$.
\end{abstract}

\section{Introduction}
\label{sec:intro}
A planner must open facilities to serve $n$ agents located in a metric
space.  Each agent reports a location, the planner opens the facilities as a
function of the reports, and an agent's cost is her distance to the nearest facility.
The planner wishes to minimize the social cost, the sum of the agent
costs, but cannot verify the reports and cannot use monetary transfers.  A
mechanism is \emph{strategyproof} if no agent ever benefits from misreporting,
and $\ratio$-approximate if its social cost is always within a factor $\ratio$ of the offline optimum.
Following \citet{PT-09}, the object of study
is to determine the best achievable approximation ratio of a strategyproof mechanism.

For the one-facility game, the median mechanism is strategyproof and optimal on
the line $\reals$ \citep{moulin1980strategy,PT-09}. 
In the plane $\reals^{2}$, the tight deterministic approximation ratio is $\sqrt{2}$ \citep{bespamyatnikh2000mobile,meir2019strategyproof,goel2023optimality}. 
Recently, \citet{barak2026facility} showed that this deterministic $\sqrt{2}$ ratio in $\reals^{2}$ can be improved to $4/\pi\approx 1.273$ by randomization.
In higher-dimensional Euclidean spaces
$\reals^{\dimension}$, the approximation ratio of the coordinate-wise median
was first bounded by $\sqrt{\dimension}$ \citep{meir2019strategyproof} and later
improved to $1.547$ \citep{gravin2025approximation}. Beyond Euclidean spaces, \citet{alon2009strategyproof} studied strategyproof one-facility location on network domains, including circles and trees. 

The two-facility game behaves quite differently.  On the line $\reals$,
\citet{PT-09} proposed the simple Two-Extremes mechanism, which opens one
facility at the leftmost reported location and the other at the rightmost
reported location.  This mechanism has approximation ratio $n-2$, and is
optimal among all deterministic strategyproof mechanisms \citep{lu2010asymptotically,fotakis2013power}.  \citet{lu2010asymptotically}
showed that randomization can break this linear barrier: their
\emph{Proportional mechanism} achieves approximation ratio $4$ on every metric
space.  The mechanism first selects an anchor agent uniformly at random, then
selects a second agent with probability proportional to her distance from the
anchor, and opens facilities at the two selected agents' locations.  On the
negative side, \citet{lu2009tighter} proved a lower bound of $1.045$ for
randomized strategyproof mechanisms on the line.  Closing the gap between
$1.045$ and $4$ in the line $\reals$ is the first open problem listed by
\citet{lu2010asymptotically}.
Neither endpoint had improved since, even on the line, let alone in higher-dimensional Euclidean spaces.

\subsection{Our contribution and techniques}
This paper improves both the upper and lower bounds for the two-facility game. 
In particular, we break the barrier of $4$ for the first time. 
We give a randomized strategyproof mechanism with approximation ratio $11/3\approx 3.667$ on every \emph{Ptolemaic} metric space --- a rich class that contains every Euclidean space, every Hilbert space, every $\mathrm{CAT}(0)$ space and every metric tree \citep{foertsch2007ptolemy}. 
This is achieved by a mixture of the Proportional mechanism and a simple mechanism we call the \emph{Global Pair} mechanism, which selects an unordered pair of agents with probability proportional to their distance and opens facilities at their locations.
On the other hand, we raise the lower bound to $(1+\sqrt{2})/2\approx 1.207$ by a two-profile instance in $\reals$.

Recall that the Proportional mechanism draws an anchor
agent uniformly at random and then draws a second agent with probability
proportional to her distance from the anchor. 
There is an asymmetry between the two draws: the first agent is drawn uniformly, while the second is drawn with probability proportional to distance from the first. 
Our proposed Global Pair mechanism removes the asymmetry between the two draws.

\begin{informal}[Global Pair mechanism; \Cref{def:gp} and \Cref{thm:gp-sp}]
The \emph{Global Pair} mechanism opens facilities at the reported locations of an unordered
pair $\set{i,j}$ drawn with probability proportional to $\dst(\loc{i},\loc{j})$.
It is strategyproof on every Ptolemaic metric space, in
particular on every Euclidean space.
\end{informal}

Taken alone, the Global Pair mechanism cannot give any improvement: like the Proportional
mechanism, its exact approximation ratio is $4$.
What makes the it useful is that the two mechanisms attain factor $4$ on
\emph{complementary} worst-case instances.  
The Proportional mechanism is tight
only when the optimal solution of the instance is highly concentrated, with almost all agents
incurring essentially zero cost.  By contrast, the Global Pair mechanism is
tight only when the optimal cost is spread out across many agents.  
Moreover, on the worst instances for either mechanism, the other achieves ratio at most $3$. We capture this complementarity by a instance-dependent parameter that measures the spread of the instance relative to a chosen optimal facility pair. 
In particular, in a metric space $(\dom,d)$, for an agent location profile
$\prof\in\dom^n$, fix an optimal facility location, let $\ocost{i}$ be the distance from agent $i$ to her nearest facility, and let $\sep$ be the distance between the two facilities. Then the \emph{dispersion} of $\prof$, denoted by $\dispersion(\prof)$, is defined as
\begin{align*}
  \dispersion(\prof)\deq
  \dfrac
  {\sum_{i\in\agents}\ocost{i}\cdot{\min\set{\ocost{i}/\sep,1}}}
  {\sum_{i\in\agents}\ocost{i}}.
\end{align*} 
The dispersion is always in $[0,1]$, and it is small when the agents are concentrated near one of the two facilities, and large when the agents are spread out across both facilities. Next we show that the approximation ratios of the two mechanisms are monotone in the dispersion, but in opposite directions, as stated in the following informal theorem.

\begin{informal}[\Cref{thm:gp-interp,thm:prop-interp}]
On every location profile $\prof\in\dom^n$, the Global Pair mechanism has approximation ratio at most $3+\dispersion(\prof)$, while the Proportional mechanism has approximation ratio at most $4-\dispersion(\prof)/2$. 
\end{informal}

This characterization of the approximation ratio by the dispersion may be of independent interest in other variants of the facility location games.
Then we randomly mix the two mechanisms, running the Proportional mechanism with probability $2/3$ and the Global Pair mechanism with probability $1/3$, which immediately gives the following upper bound, which is main result of this paper.

\begin{informal}[\Cref{thm:mixture-exact}]
Let $\Mix{2/3}$ be the mechanism that runs the Proportional mechanism with probability $2/3$ and the Global Pair mechanism with probability $1/3$. 
Then $\Mix{2/3}$ is strategyproof on every Ptolemaic metric space and has approximation ratio exactly $11/3$.
\end{informal}

This is the first strategyproof two-facility mechanism with ratio below $4$.
The approximation ratio analysis actually only requires the triangle inequality, so it holds on every metric space. The truthfulness holds on every Ptolemaic metric space, which includes every Euclidean space. In \Cref{app:gp-counterexample}, we give a counterexample showing that the Global Pair mechanism is not strategyproof on every metric space, and hence the mixture is not strategyproof on every metric space either.

Our second contribution is a new lower bound of $(1+\sqrt{2})/2\approx 1.207$ for the two-facility game, which improves the previous lower bound of $1.045$ by \citet{lu2009tighter}, as stated in the following informal theorem.

\begin{informal}[\Cref{thm:lb-asymptotic}]
Any (randomized) strategyproof mechanism for the two-facility game has approximation ratio at least $(1+\sqrt{2})/2$.
\end{informal}

\begin{figure}[H]
  \centering
  \newcommand{\lbcoordlabelvsep}{9.5pt}
  \newcommand{\lbblockdiameter}{8mm}
  \newcommand{\lbblocklabelscale}{0.8}
  \newcommand{\lbblocknode}[3]{%
    \node[#1circle] at (#2) {};
    \node[blocktext] at (#2) {\scalebox{\lbblocklabelscale}{#3}};
  }
  \begin{subfigure}[t]{0.47\textwidth}
    \centering
    \begin{tikzpicture}[
        x=1cm,
        y=1cm,
        axis/.style={-{Latex[length=1.5mm]}, draw=black!55,
          line width=0.45pt},
        coordlabel/.style={font=\scriptsize, below=\lbcoordlabelvsep},
        agent/.style={circle, draw=black, fill=black, inner sep=1.6pt},
        blockcircle/.style={circle, draw=black!70, fill=black!8,
          minimum size=\lbblockdiameter, inner sep=0pt},
        blocktext/.style={align=center, inner sep=0pt},
        deviator/.style={circle, draw=black, fill=black,
          inner sep=1.6pt},
        misreport/.style={-{Latex[length=2mm]}, draw=black,
          dashed, line width=0.8pt},
        every node/.style={font=\small}
      ]
      \node[anchor=east] at (-0.10,1.45) {$\prof$};
      \node[anchor=east] at (-0.10,0.10) {$\prof'$};

      \draw[axis] (0,1.45) -- (5.65,1.45);
      \draw[axis] (0,0.10) -- (5.65,0.10);

      \node[agent] at (0.55,1.45) {};
      \lbblocknode{block}{2.70,1.45}{$n-2$}
      \node[deviator] (lu-old) at (4.75,1.45) {};
      \node[coordlabel] at (0.55,1.45) {$-1$};
      \node[coordlabel] at (2.70,1.45) {$0$};
      \node[coordlabel] at (4.75,1.45) {$1$};

      \node[agent] at (0.55,0.10) {};
      \lbblocknode{block}{2.70,0.10}{$n-2$}
      \node[deviator] (lu-new) at (5.30,0.10) {};
      \node[coordlabel] at (0.55,0.10) {$-1$};
      \node[coordlabel] at (2.70,0.10) {$0$};
      \node[coordlabel] at (5.30,0.10) {$1+\alpha$};

      \draw[misreport] (lu-old) to[out=0,in=70] (lu-new);
    \end{tikzpicture}
    \caption{The lower bound construction in \citet{lu2009tighter}.}
    \label{fig:lb-comparison-lu}
  \end{subfigure}
  \hfill
  \begin{subfigure}[t]{0.49\textwidth}
    \centering
    \begin{tikzpicture}[
        x=1cm,
        y=1cm,
        axis/.style={-{Latex[length=1.5mm]}, draw=black,
          line width=0.45pt},
        coordlabel/.style={font=\scriptsize, below=\lbcoordlabelvsep},
        blockcircle/.style={circle, draw=black!70, fill=black!8,
          minimum size=\lbblockdiameter, inner sep=0pt},
        midblockcircle/.style={circle, draw=black, fill=black!10,
          minimum size=\lbblockdiameter, inner sep=0pt},
        blocktext/.style={align=center, inner sep=0pt},
        deviator/.style={circle, draw=black, fill=black,
          inner sep=1.6pt},
        misreport/.style={-{Latex[length=2mm]}, draw=black,
          dashed, line width=0.8pt},
        every node/.style={font=\small}
      ]
      \node[anchor=east] at (-0.10,1.45) {$\prof$};
      \node[anchor=east] at (-0.10,0.10) {$\prof'$};

      \draw[axis] (0,1.45) -- (5.80,1.45);
      \draw[axis] (0,0.10) -- (5.80,0.10);
      \draw[white, line width=2.2pt] (0.90,1.34) -- (1.08,1.56);
      \draw[black!55, line width=0.45pt] (0.88,1.34) -- (0.98,1.56);
      \draw[black!55, line width=0.45pt] (1.00,1.34) -- (1.10,1.56);
      \draw[white, line width=2.2pt] (0.90,-0.01) -- (1.08,0.21);
      \draw[black!55, line width=0.45pt] (0.88,-0.01) -- (0.98,0.21);
      \draw[black!55, line width=0.45pt] (1.00,-0.01) -- (1.10,0.21);

      \lbblocknode{block}{1.65,1.45}{$m$}
      \lbblocknode{midblock}{3.40,1.45}{$k$}
      \node[deviator] (our-old) at (3.22,1.45) {};
      \lbblocknode{block}{5.15,1.45}{$m$}
      \node[coordlabel] at (1.65,1.45) {$0$};
      \node[coordlabel] at (3.40,1.45) {$1/2$};
      \node[coordlabel] at (5.15,1.45) {$1$};

      \node[deviator] (our-new) at (0.35,0.10) {};
      \lbblocknode{block}{1.65,0.10}{$m$}
      \lbblocknode{midblock}{3.40,0.10}{$k-1$}
      \lbblocknode{block}{5.15,0.10}{$m$}
      \node[coordlabel] at (0.35,0.10) {$-m$};
      \node[coordlabel] at (1.65,0.10) {$0$};
      \node[coordlabel] at (3.40,0.10) {$1/2$};
      \node[coordlabel] at (5.15,0.10) {$1$};

      \draw[misreport]
        (our-old.south) .. controls (3.35,0.68) and (0.35,0.68)
        .. (our-new.north);
    \end{tikzpicture}
    \caption{The lower bound construction in this paper.}
    \label{fig:lb-comparison-ours}
  \end{subfigure}
  \caption{Comparison of the lower-bound constructions. The black dot denotes a
    single agent, while a large circle represents a co-located block of agents,
    and the number inside the circle is the block size. The dashed arrow marks
    the designated agent's misreport.}
  \label{fig:lower-bound-comparison}
\end{figure}
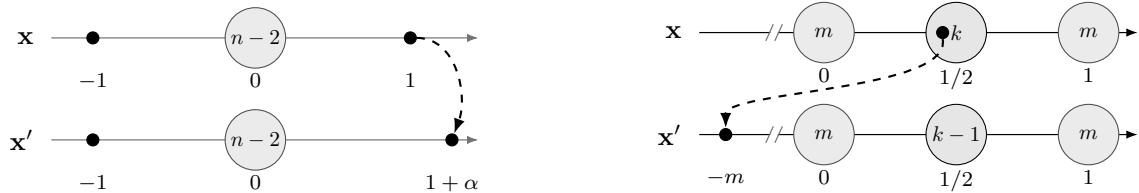
 
The key new ingredient is a block-amplification construction; see
\Cref{fig:lower-bound-comparison}.  The construction of
\citet{lu2009tighter} perturbs a single isolated agent, so strategyproofness
directly constrains only that agent's expected cost.  We instead place $m$
agents at each endpoint and a block of $k$ agents at the midpoint, and let one
midpoint agent misreport to a distant outlier.  The two profiles have different
optimal solutions.  Before the misreport, placing one facility at each endpoint
gives optimal cost $k/2$, paid by the midpoint block.  After the misreport,
placing one facility at the outlier and the other at the midpoint gives optimal
cost $m$, paid by the two endpoint blocks.
$k-1$ agents remaining at the midpoint incurs this same post-misreport cost.
The key is that strategyproofness requires the deviating agent's expected cost at
her true midpoint location becomes weakly larger after the misreport.  Every one of the
$k-1$ agents remaining at the midpoint incurs this same post-misreport cost.
Thus one incentive constraint is amplified by the size of the midpoint block.
Optimizing the relative block sizes gives $k/m\to2-\sqrt{2}$ and the lower
bound $(1+\sqrt{2})/2$.

\section{Preliminaries}
\label{sec:prelim}
Let $\metricspace$ be a metric space, where $\dom$ is the ambient space and $\dst\colon\dom\times\dom\to[0,\infty)$ is a metric satisfying three axioms:
\begin{align*}
  \dst(x,y)=\dst(y,x),\qquad
  \dst(x,y)=0\iff x=y,\qquad
  \dst(x,y)\le\dst(x,z)+\dst(z,y),
\end{align*}
for all $x,y,z\in\dom$.  
The $\dimension$-dimensional \emph{Euclidean space} $\reals^{\dimension}$ is the metric space with the ambient space $\dom=\reals^{\dimension}$ and the metric $\dst(x,y)=\norm{x-y}_{2}$. 
Next we introduce the Ptolemaic metric space, which is a rich class of metric spaces that contains every Euclidean space, every Hilbert space, every $\mathrm{CAT}(0)$ space and every metric tree \citep{foertsch2007ptolemy}.
\begin{definition}[Ptolemaic space]\label{def:ptolemaic}
A metric space $\metricspace$ is \emph{Ptolemaic} if every four points $x,y,z,w\in \dom$ satisfy the \emph{Ptolemy inequality}:
\begin{align}\label{eq:ptolemy}
  \dst(x,z)\cdot\dst(y,w)\le\dst(x,y)\cdot\dst(z,w)+\dst(x,w)\cdot\dst(y,z).
\end{align}
\end{definition}

Let $\agents=\set{1,\dots,n}$ be the set of agents.  The location of agent
$i\in\agents$ is $\loc{i}\in\dom$, and the \emph{location profile} is the vector
$\prof=\lranglefix{\loc{1},\dots,\loc{n}}\in\dom^{n}$. 
We write $\prof=\lranglefix{\loc{i},\prof_{-i}}$ when we single out agent
$i$, and $\lranglefix{\prof_{S},\prof_{-S}}$ for a subset $S\subseteq\agents$.

In the \emph{two-facility game}, a \emph{randomized mechanism} is a function
$\Mech\colon\dom^{n}\to\pairs{\dom}$ that maps every location profile to a distribution over unordered pairs of facility locations, where $\pairs{\dom}$ is the set of distributions over pairs of facility locations.
A \emph{facility pair} is an unordered pair
$\pair=(\ell_{1},\ell_{2})\in\dom^{2}$.
Given agent location profile $\prof$ and the mechanism $\Mech$, 
the cost of an agent $i$ is her expected distance to
the nearer of the two facilities:
\begin{align*}
  \cost\parent{\Mech(\prof),\loc{i}}
  \deq\expect[\pair\sim\Mech(\prof)]
    {\min\set{\dst(\ell_{1},\loc{i}),\dst(\ell_{2},\loc{i})}}.
\end{align*}
For simplicity we write $\Mech(\loc{i},\prof_{-i})$ for
$\Mech\parent{\lranglefix{\loc{i},\prof_{-i}}}$, and
$\Mech(\prof_{S},\prof_{-S})$ for
$\Mech\parent{\lranglefix{\prof_{S},\prof_{-S}}}$.

The \emph{social cost} of a mechanism $\Mech$ on a location profile $\prof$
is the total cost of all $n$ agents, denoted $\SC(\Mech,\prof)$ and defined as
\begin{align*}
  \SC(\Mech,\prof)\deq\sum_{i\in\agents}\cost\parent{\Mech(\prof),\loc{i}} .
\end{align*}
For a location profile $\prof$, let $\OPT(\prof)$ denote the minimum social
cost achievable by any pair of facilities on $\dom^2$.  We say that a mechanism
$\Mech$ has \emph{approximation ratio} $\ratio$ if for every profile $\prof\in \dom^{n}$,
\begin{align*}
  \SC(\Mech,\prof)\;\le\;\ratio\OPT(\prof),
\end{align*}
and there exists a profile $\prof$ such that the $\le$ holds with equality.

Now we define the notion of strategyproofness, which is the standard incentive property required for mechanisms.  A mechanism is strategyproof if no agent can reduce her expected cost by misreporting her location, regardless of the reports of other agents.

\begin{definition}[Strategyproofness]\label{def:sp}
A mechanism $\Mech$ is \emph{strategyproof} if for every agent $i\in \agents$, every location profile
$\prof=\lranglefix{\loc{i},\prof_{-i}}\in\dom^{n}$, every misreported location   
$\loc{i}\primed\in\dom$, it holds that
\begin{align*}
  \cost\parent{\Mech(\loc{i},\prof_{-i}),\loc{i}}
  \;\le\;
  \cost\parent{\Mech(\loc{i}\primed,\prof_{-i}),\loc{i}}.
\end{align*}
\end{definition}

We next recall the Proportional mechanism of \citet{lu2010asymptotically}, which is strategyproof on every metric space, and has approximation ratio exactly $4$.

\begin{definition}[Proportional mechanism]\label{def:prop}
The \emph{Proportional} mechanism $\Prop$ draws an anchor agent $i$ uniformly from $\agents$ and,
conditionally on $i$, draws the second agent $j$ from $\agents$ with probability
\begin{align*}
  \prob[\Prop]{j\condition i}=\frac{\dst(\loc{i},\loc{j})}{\sum_{k\in\agents}\dst(\loc{i},\loc{k})},
\end{align*}
opening facilities at $\loc{i}$ and $\loc{j}$. If all reported locations coincide, the Proportional mechanism opens
both facilities at the common reported location.
\end{definition}

Next we introduce the key technical ingredient of our paper, the
\emph{dispersion} of a location profile.  Relative to a chosen optimal facility
pair, it measures how far the agents lie from their assigned facilities
compared with the distance between those facilities. 
\begin{definition}[Dispersion]\label{def:dispersion}
Let $\prof\in\dom^{n}$ be a location profile whose optimum is attained.  
Fix an optimal facility pair
$(\ofac{1},\ofac{2})$ and let
$\sep=\dst(\ofac{1},\ofac{2})>0$ be the distance between the two facilities.
For each agent $i\in\agents$, define her optimal cost under the chosen optimal facility pair by $\ocost{i}\deq\min\set{\dst(\loc{i},\ofac{1}),\dst(\loc{i},\ofac{2})}$.
Then the \emph{dispersion} $\dispersion(\prof)$ of $\prof$ is
\begin{align*}
  \dispersion(\prof)\deq
  \dfrac
  {\sum_{i\in\agents}\ocost{i}\cdot{\min\set{\ocost{i}/\sep,1}}}
  {\sum_{i\in\agents}\ocost{i}}
\end{align*}
\end{definition}

The dispersion $\dispersion(\prof)$ is a
weighted average of the capped ratios
$\min\set{\ocost{i}/\sep,1}$, with weights equal to $\ocost{i}$.
Since $\min\set{\ocost{i}/\sep,1}\le1$, we have $0\le\dispersion(\prof)\le 1$.
When $\dispersion(\prof)$ approaches $0$, we have that $\ocost{i}\ll\sep$ for most agents, and they are therefore concentrated around their assigned
optimal facilities.  
Conversely, $\dispersion(\prof)$ close to $1$ means that most agents are far from their assigned optimal facilities, and the optimal solution is therefore dispersed. 
\section{The Global Pair Mechanism and Its Mixture with the Proportional Mechanism}
\label{sec:mech}
In this section, we introduce the Global Pair mechanism, the main mechanism of
the paper, and prove that it is strategyproof on every Ptolemaic metric space;
see \Cref{def:gp} and \Cref{sec:truthful-of-global-pair}.  On its own, however, the Global Pair
mechanism has approximation ratio $4$, exactly equal to that of the
Proportional mechanism. 
By expressing the profile-dependent approximation
bounds of both mechanisms in terms of the dispersion parameter
$\dispersion(\prof)$ (see \Cref{def:dispersion}), we show the most interesting part of this paper: that the approximation ratio of Global Pair mechanism is upper bounded by $3+\dispersion(\prof)$, whereas the approximation ratio of Proportional mechanism is bounded by $4-\dispersion(\prof)/2$.  The former increases linearly with $\dispersion(\prof)$, while the latter decreases linearly with $\dispersion(\prof)$; see \Cref{sec:interp}.  The two mechanisms are therefore complementary.  A randomized mixture of them achieves approximation ratio $11/3$; see \Cref{thm:mixture-exact}.
. The two
mechanisms are therefore complementary.  A randomized mixture of them breaks
the factor-$4$ barrier: running the Proportional mechanism with probability
$2/3$ and the Global Pair mechanism with probability $1/3$ yields a
strategyproof mechanism with approximation ratio exactly $11/3$; see
\Cref{thm:mixture-exact}.

\begin{definition}[Global Pair mechanism]\label{def:gp}
The \emph{Global Pair} mechanism $\GPair$ draws an unordered agent pair $\setfix{i,j}$
with probability
\begin{align*}
  \prob[\GPair]{\set{i,j}}
  =\frac{\dst(\loc{i},\loc{j})}{\sum_{k<l}\dst(\loc{k},\loc{l})},
\end{align*}
opening facilities at $\loc{i}$ and $\loc{j}$. If all reported locations coincide, the Global Pair mechanism opens
both facilities at the common reported location.
\end{definition}

\xhdr{Comparison of the two mechanisms.}
Define $\pw\deq\sum_{k<l}\dst(\loc{k},\loc{l})$ as the \emph{total distance} over all
unordered pairs of agents, the normalizing constant of \Cref{def:gp}, and
$\degr{i}\deq\sum_{j\in\agents}\dst(\loc{i},\loc{j})$ for the total distance from
agent $i$ to all agents, her \emph{distance degree}.
To compare the two mechanisms, realize the unordered draw of \Cref{def:gp} as an
ordered one: for all $i,j\in\agents$ let
$
\prob[\GPair]{i,j}\deq{\dst(\loc{i},\loc{j})}/{2\pw}
$
be the probability that the first selected agent is $i$ and the second is $j$.  
Then, the marginal probability that $i$ is the first selected agent is
\begin{align*}
  \prob[\GPair]{i}=\sum_{j\in\agents}\prob[\GPair]{i,j}
  =\frac{1}{2\pw}\sum_{j\in\agents}\dst(\loc{i},\loc{j})
  =\frac{\degr{i}}{2\pw},
\end{align*}
Given agent $i$ is the first selection, the conditional probability that $j$ is the second selection is 
\begin{align*}
  \prob[\GPair]{j\condition i}=\frac{\prob[\GPair]{i,j}}{\prob[\GPair]{i}}
  =\frac{\dst(\loc{i},\loc{j})/(2\pw)}{\degr{i}/(2\pw)}
  =\prob[\Prop]{j\condition i} ,
\end{align*}
the last equality being the second draw of \Cref{def:prop}.  
Therefore, the Global Pair
mechanism may equivalently be described as first drawing an anchor $i$ with
probability $\degr{i}/(2\pw)$ and then applying the same distance-proportional
second draw as the Proportional mechanism.  
The two differ only in the
anchor distribution:
that of the Proportional mechanism is uniform, that of the Global Pair mechanism
is distance-degree biased.  This is the sense in which they are
complementary, and it is what \Cref{thm:gp-interp,thm:prop-interp} quantify. Next we define the mixture of the two mechanisms, which is a randomization between them.

\begin{definition}[Mixture of $\Prop$ and $\GPair$]\label{def:mix}
For a constant $\lambda\in[0,1]$, the \emph{mixture} $\Mix{\lambda}$ is the randomized
mechanism that runs the Proportional mechanism $\Prop$ with probability $\lambda$ and
the Global Pair mechanism $\GPair$ with probability $1-\lambda$.
\end{definition}

\subsection{Truthfulness of the Global Pair mechanism}
\label{sec:truthful-of-global-pair}

We show that the Global Pair mechanism is strategyproof on a
large and natural subclass named Ptolemaic spaces, which includes every Euclidean space, Hilbert spaces and metric trees \citep{foertsch2007ptolemy}.

\begin{theorem}[Truthfulness of the Global Pair mechanism]\label{thm:gp-sp}
The Global Pair mechanism is strategyproof on every Ptolemaic metric space.
\end{theorem}

We give a sketch of the proof of \Cref{thm:gp-sp} 
, since this explains
where the Ptolemy inequality \cref{eq:ptolemy} is needed.  Fix an agent with
true location $x$ who reports $x\primed$, let $z_{1},\dots,z_{m}$ be the reports
of the other agents, and abbreviate
\begin{align*}
  \devd=\dst(x,x\primed),\qquad
  \tdist{i}=\dst(x,z_{i}),\qquad
  \rdist{i}=\dst(x\primed,z_{i}),\qquad
  \odist{i}{j}=\dst(z_{i},z_{j}).
\end{align*}
Let $\Ssum=\sum_{i}\tdist{i}$ and $\Tsum=\sum_{i}\rdist{i}$ be the total distances from her true
and from her reported location to the other reports, and let
$\pw=\sum_{i<j}\odist{i}{j}$ be the total distance over pairs of other reports.  
By some computation, we get the agent's expected costs when truthful and when deviating are therefore $\Aterm/(\pw+\Ssum)$ and $(\Aterm+\Bterm)/(\pw+\Tsum)$, respectively. So
truthfulness is equivalent to $\Aterm(\Tsum-\Ssum)\le \Bterm(\pw+\Ssum)$, which is implied by the stronger
inequality $\Aterm(\Tsum-\Ssum)\le \Bterm\pw$.  Expanding both sides we get
\begin{align}
  \label{eq:truthful_sketch}
  \Bterm\pw-\Aterm(\Tsum-\Ssum)
  =\sum_{j<k}\underbrace{\odist{j}{k}\sum_{i}
     \underbrace{\Bigl[\rdist{i}\cdot\min\set{\devd,\tdist{i}}
       -\min\set{\tdist{j},\tdist{k}}\cdot(\rdist{i}-\tdist{i})\Bigr]
     }_{inner\ term}}_{outer\ term}.
\end{align}
For $i\in\set{j,k}$, the inner terms are nonnegative by the triangle inequality and elementary algebra.
For $i\notin\set{j,k}$, consider outer terms.
Each such term is
determined by an ordered triple $\lranglefix{j,k,i}$.
For given distinct $j,k,i$, there are three outer terms, corresponding to the three possible ordered triples $\lranglefix{j,k,i}$, $\lranglefix{i,j,k}$ and $\lranglefix{k,i,j}$. 
Label the three indices so that $\tdist{i}\le\tdist{j}\le\tdist{k}$.  
So the sum of the three outer terms becomes
\begin{align*}
  \odist{j}{k}[\rdist{i}\min\set{\devd,\tdist{i}}
    -\tdist{j}(\rdist{i}-\tdist{i})]
  +\odist{i}{k}[\rdist{j}\min\set{\devd,\tdist{j}}
    -\tdist{i}(\rdist{j}-\tdist{j})]
  +\odist{i}{j}[\rdist{k}\min\set{\devd,\tdist{k}}
    -\tdist{i}(\rdist{k}-\tdist{k})].
\end{align*}
By triangle inequality and elementary algebra, 
the nonnegativity of above sum is implied by
$\rdist{i}\odist{j}{k}\le\rdist{j}\odist{i}{k}+\rdist{k}\odist{i}{j}$, which exactly is the Ptolemy inequality in \cref{eq:ptolemy} for the four points $x\primed,z_{i},z_{j},z_{k}$.  
Thus, in every Ptolemaic metric space, the Global Pair mechanism is strategyproof.
The formal proof is in \Cref{app:gp-sp}. 

We note that the Ptolemaic assumption in \Cref{thm:gp-sp} cannot be removed in general:
the Global Pair mechanism is not strategyproof on certain metric
spaces.
\Cref{app:gp-counterexample} exhibits a profile of $14$ agents on the unit-edge
shortest-path metric of $K_{2,3}$, which is not Ptolemaic, at which a single
misreport strictly lowers the deviator's expected cost.

Next, by the truthfulness of the Proportional mechanism in every metric space, we immediately get the following corollary.
\begin{corollary}\label{cor:mix-sp}
For every constant $\lambda\in[0,1]$, the mechanism $\Mix{\lambda}$ is strategyproof on every Ptolemaic metric space.
\end{corollary}

\subsection{Approximation analysis by dispersion and breaking the barrier by mixture}\label{sec:interp}
We now turn from incentives to approximation, and analyze the social costs of
the Global Pair mechanism $\GPair$ and the Proportional mechanism $\Prop$ on a
fixed location profile $\prof$.  We bound both in terms of the dispersion
parameter $\dispersion(\prof)$ of \Cref{def:dispersion}.  The two bounds move in opposite
directions: the social cost of $\GPair$ is at most $(3+\dispersion(\prof))\cdot\OPT(\prof)$, which
degrades as the optimal solution spreads out, while that of $\Prop$ is at most
$(4-\dispersion(\prof)/2)\cdot\OPT(\prof)$, which degrades as it concentrates
(\Cref{thm:gp-interp,thm:prop-interp}).  Mixing the two mechanisms therefore
cancels the dependence on $\dispersion$, and at $\lambda=2/3$ it yields approximation
ratio exactly $11/3$ (\Cref{thm:mixture-exact}).

We first present parametric versions of the approximation bounds for the two mechanisms, which are the main technical results of this section. 
\begin{theorem}\label{thm:gp-interp} 
For a location profile $\prof$ in a metric space $\metricspace$, suppose that
the optimal social cost $\OPT(\prof)$ is attained, and let $\dispersion(\prof)$ be the
dispersion parameter defined in \Cref{def:dispersion}.
Under truthful reporting, the expected social
cost of the Global Pair mechanism $\SC(\GPair,\prof)$ satisfies
\begin{align*}
  \SC(\GPair,\prof)\le(3+\dispersion(\prof))\cdot\OPT(\prof).
\end{align*}
Moreover, the approximation ratio of the Global Pair mechanism is exactly $4$.
\end{theorem}

\begin{proof}
If $\OPT(\prof)=0$, the profile has at most two occupied locations and the
Global Pair mechanism incurs zero social cost.  Hence assume
$\OPT(\prof)>0$.
Fix an optimal facility pair $(\ofac{1},\ofac{2})$.  For each agent $i$, let
$\ocost{i}$ be her distance to the nearer optimal facility.  Thus
$\Rsum=\OPT(\prof)=\sum_i\ocost{i}$.  Recall that
$\sep=\dst(\ofac{1},\ofac{2})$ is the distance between the optimal facilities,
and let
$
  \Qsum
  =\sum_{i\in\agents}\ocost{i}\min\set{\ocost{i},\sep},
$
and 
$
  \dispersion
  =\frac{\Qsum}{\sep\Rsum}.
$
Define  
$
  G\deq\SC(\GPair,\prof)
$ and 
$
  \pw\deq\sum_{a<b}\dst(\loc{a},\loc{b}).
$
Here $G$ is the expected social cost of the Global Pair mechanism and $\pw$ is
the total pairwise distance, which normalizes its pair-selection probabilities.

For a pair $\set{a,b}$, let
\begin{align*}
  C_{ab}
  \deq\sum_{k\in\agents}
    \min\set{\dst(\loc{k},\loc{a}),\dst(\loc{k},\loc{b})}
\end{align*}
be the social cost when facilities are opened at $\loc{a}$ and $\loc{b}$.
Since the Global Pair mechanism selects $\set{a,b}$ with probability
$\dst(\loc{a},\loc{b})/\pw$, its expected social cost is
\begin{align*}
  G
  =\sum_{a<b}\frac{\dst(\loc{a},\loc{b})}{\pw}\,C_{ab}
  =\frac{1}{\pw}
    \underbrace{\sum_{a<b}\dst(\loc{a},\loc{b})C_{ab}}
      _{\text{unnormalized numerator}}.
\end{align*}
Thus $\pw G$ is the numerator of the expectation: it is the sum, over all
possible selected pairs, of the pair's weight times the social cost generated
by that pair.

The two agents in the selected pair incur zero cost.  Therefore every term in
$\pw G$ can be grouped with an unordered triple $\set{i,j,k}$.  Write
$d_{ab}=\dst(\loc{a},\loc{b})$ and relabel the three distances in the triple as
$\alpha\le\beta\le\gamma$.  The triple's contribution is
\begin{align}\label{eq:triple-weight}
  \Phi_{ijk}
  \deq d_{ij}\min\set{d_{ik},d_{jk}}
       +d_{ik}\min\set{d_{ij},d_{jk}}
       +d_{jk}\min\set{d_{ij},d_{ik}}
  =2\alpha\beta+\alpha\gamma.
\end{align}
Indeed, each of the three terms corresponds to selecting one edge of the
triple; that edge supplies the selection weight, and the remaining agent pays
her distance to its nearer endpoint.

The desired inequality is homogeneous in all distances.  We may therefore
scale the metric by $1/\sep$ and assume $\sep=1$.  Define
$H(c)=3c+c\min\set{c,1}$.  The auxiliary inequality proved in
\Cref{app:triple} gives
\begin{align*}
  \Phi_{ijk}\;\le\;\dst(\loc{j},\loc{k})\cdot H(\ocost{i})
   +\dst(\loc{i},\loc{k})\cdot H(\ocost{j})
   +\dst(\loc{i},\loc{j})\cdot H(\ocost{k}).
\end{align*}
Summing this inequality over all unordered triples gives
\begin{align*}
  \pw G
  =\sum_{i<j<k}\Phi_{ijk} 
  \le\sum_{i\in\agents}H(\ocost{i})
    \sum_{\substack{j<k\\ j,k\ne i}}\dst(\loc{j},\loc{k})
  =\sum_{i\in\agents}H(\ocost{i})\bigl(\pw-\degr{i}\bigr)
  \le\pw\sum_{i\in\agents}H(\ocost{i}).
\end{align*}
Here $\degr{i}=\sum_j\dst(\loc{i},\loc{j})$ is the distance degree of agent
$i$, namely the total weight of all pairs incident to $i$.  Consequently,
$\pw-\degr{i}$ is exactly the total weight of pairs among the other agents,
which explains the last inequality.

Dividing by $\pw$ and using $\sep=1$ now yields
\begin{align*}
  G
  \le\sum_{i\in\agents}H(\ocost{i})
  =3\sum_{i\in\agents}\ocost{i}
    +\sum_{i\in\agents}\ocost{i}\min\set{\ocost{i},1} 
  =3\Rsum+\Qsum
   =(3+\dispersion)\Rsum.
\end{align*}
Finally, rescaling the metric multiplies both $G$ and $\Rsum$ by the same
factor and leaves $\dispersion$ unchanged.  Hence the inequality also holds in the
original scale.
\end{proof}

\begin{theorem}\label{thm:prop-interp}
For a location profile $\prof$ in a metric space $\metricspace$, suppose that
the optimal social cost $\OPT(\prof)$ is attained, and let $\dispersion(\prof)$ be the
dispersion parameter defined in \Cref{def:dispersion}.  Under truthful reporting,
the expected social cost of the Proportional mechanism satisfies
\begin{align*}
  \SC(\Prop,\prof)
  \le\parentfix{4-{\dispersion(\prof)}/{2}}\cdot\OPT(\prof).
\end{align*}
\end{theorem}

\begin{proof}[Proof sketch]
The proof of \Cref{thm:prop-interp} refines the factor-$4$ analysis of
\citet{lu2010asymptotically}.  Like their proof, it conditions on the uniformly
selected anchor and partitions the agents according to an optimal two-facility
solution.  The original proof separately bounds the conditional costs of the
agents in the anchor's optimal cluster and those in the other optimal cluster,
which yields the factor $4$.  We instead keep these costs together and, after
averaging over all possible anchors, write the bound as $4\cdot\OPT(\prof)$ minus an
explicit nonnegative remainder.  We show that this remainder is at least
$\dispersion(\prof)\cdot\OPT(\prof)/2$, which yields the profile-dependent factor
$4-\dispersion(\prof)/2$.

The full proof is in \Cref{app:prop-interp}; we sketch the main idea here, which
produces the improvement over $4$.  Fix an optimal facility pair and its
associated clusters $\clus{1}$ and $\clus{2}$.  Recall that
$\nclus{1}=\setsize{\clus{1}}$, $\nclus{2}=\setsize{\clus{2}}$, and $n=\nclus{1}+\nclus{2}$.  For
$r\in\set{1,2}$, define the optimal cost and capped second moment of cluster
$\clus{r}$ by
\begin{align*}
  \Rsum_r
  &\deq\sum_{i\in\clus{r}}\ocost{i},
  &
  \Qsum_r
  &\deq\sum_{i\in\clus{r}}
    \ocost{i}\min\set{\ocost{i},\sep}.
\end{align*}
Thus $\Rsum=\Rsum_1+\Rsum_2$.  For this proof, set
$\Qsum\deq\Qsum_1+\Qsum_2$.

Consider an anchor $i\in\clus{1}$.  Let
\begin{align*}
  U_i
  &\deq\sum_{k\in\clus{1}}\dst(\loc{i},\loc{k}),
  &
  V_i
  &\deq\sum_{k\in\clus{2}}\dst(\loc{i},\loc{k}).
\end{align*}
Hence $U_i$ is the total distance from the anchor to agents in its own cluster,
whereas $V_i$ is the total distance to agents in the other cluster.  
By triangle inequality and elementary algebra, we have the the social cost of the Proportional mechanism conditioned on $i$ being the anchor is
\begin{align*}
  \expect[\Prop]{\text{social cost}\condition i\text{ is the anchor}}
  \le 2U_i+3\Rsum_2
    -\frac{U_i(U_i+3\Rsum_2)}{U_i+V_i}.
\end{align*}
The last term is nonnegative slack.  The factor-$4$ analysis of
\citet{lu2010asymptotically} bounds the costs contributed by the two optimal
clusters separately, so the negative term in the preceding display does not
appear in its final estimate.  Here, we retain this term and sum it over all
possible anchors.  Define
\begin{align*}
  \Lambda_1
  \deq\sum_{i\in\clus{1}}
    \frac{U_i(U_i+3\Rsum_2)}{U_i+V_i},
\end{align*}
and define $\Lambda_2$ symmetrically by interchanging the two clusters.  For
every $i,k\in\clus{1}$, the triangle inequality through $\ofac{1}$ gives
$\dst(\loc{i},\loc{k})\le\ocost{i}+\ocost{k}$.  Therefore,
\begin{align*}
  \sum_{i\in\clus{1}}U_i
  =\sum_{i\in\clus{1}}\sum_{k\in\clus{1}}
    \dst(\loc{i},\loc{k})
  \le\sum_{i\in\clus{1}}\sum_{k\in\clus{1}}
    \bigl(\ocost{i}+\ocost{k}\bigr)
   =2\nclus{1}\Rsum_1.
\end{align*}
The same argument gives the symmetric bound for $\clus{2}$.
Consequently, averaging the fixed-anchor bounds over the uniformly chosen
anchor yields
\begin{align*}
  \SC(\Prop,\prof)
  \le 3\Rsum+\frac{\nclus{1}\Rsum_1+\nclus{2}\Rsum_2}{n}
    -\frac{\Lambda_1+\Lambda_2}{n}.
\end{align*}

It remains to quantify the retained slack.  The key aggregate estimates are
\begin{align}\label{eq:lambda-lb}
  \Lambda_1+\nclus{2}\Rsum_1
  &\ge\frac{n\Qsum_1}{2\sep},
  &
  \Lambda_2+\nclus{1}\Rsum_2
  &\ge\frac{n\Qsum_2}{2\sep}.
\end{align}
They follow from the fact that each optimal facility is a $1$-median of its
assigned cluster, together with the triangle inequality.  Crucially, these are
aggregate rather than pointwise estimates: in the first cluster, the residual
terms cancel after summation because
$\sum_{i\in\clus{1}}(\Rsum_1/\nclus{1}-\ocost{i})=0$, and the second cluster is
symmetric.

Adding the two inequalities in \cref{eq:lambda-lb} and substituting the result
into the averaged bound gives
\begin{align*}
  \SC(\Prop,\prof)
  \le 3\Rsum+\frac{\nclus{1}\Rsum_1+\nclus{2}\Rsum_2}{n}
    +\frac{\nclus{2}\Rsum_1+\nclus{1}\Rsum_2}{n}-\frac{\Qsum}{2\sep}
  =4\Rsum-\frac{\Qsum}{2\sep}
  =\Bigl(4-\frac{\dispersion}{2}\Bigr)\Rsum,
\end{align*}
which proves the theorem.
\end{proof}

The two approximation bounds vary in opposite directions with the dispersion parameter: the Proportional bound decreases as dispersion increases, whereas the Global Pair bound increases. This complementarity yields a uniform approximation guarantee for their mixture.

\begin{theorem}\label{thm:mixture-exact}
Let $\Mix{2/3}$ be the mechanism of \Cref{def:mix} at $\lambda=2/3$, which runs
the Proportional mechanism with probability $2/3$ and the Global Pair mechanism
with probability $1/3$.  We have $\Mix{2/3}$ is strategyproof on every Ptolemaic metric space, and its approximation ratio is exactly $11/3$.
\end{theorem}

\begin{proof}
The upper bound follows immediately from \Cref{thm:gp-interp,thm:prop-interp}.
For tightness, consider the family
$\prof^{m}=\lranglefix{0^{m},\varepsilon,1}$ with $\varepsilon=m^{-2}$, for which
$\OPT(\prof^{m})=\varepsilon$, attained by facilities at $0$ and $1$.
Conditioning the Proportional mechanism on its anchor gives
\begin{align*}
  C_{0}=\frac{\varepsilon(2-\varepsilon)}{1+\varepsilon},\qquad
  C_{\varepsilon}=\frac{2m\varepsilon(1-\varepsilon)}{m\varepsilon+1-\varepsilon},\qquad
  C_{1}=\frac{m\varepsilon(2-\varepsilon)}{m+1-\varepsilon},
\end{align*}
whence $\Prop=(mC_{0}+C_{\varepsilon}+C_{1})/(m+2)$ and
\begin{align}\label{eq:family-conc}
  \frac{\Prop}{\OPT}
  =\frac{m}{m+2}\left[\frac{2-\varepsilon}{1+\varepsilon}
    +\frac{2(1-\varepsilon)}{m\varepsilon+1-\varepsilon}
    +\frac{2-\varepsilon}{m+1-\varepsilon}\right]\longrightarrow4 .
\end{align}
For the Global Pair mechanism the total pair weight is
$\pw=m(1+\varepsilon)+1-\varepsilon$
and the weighted numerator is $m\varepsilon(3-2\varepsilon)$, so
\begin{align*}
  \frac{\GPair}{\OPT}=\frac{m(3-2\varepsilon)}{m(1+\varepsilon)+1-\varepsilon}
  \longrightarrow3 .
\end{align*}
Hence $\Mix{\lambda}/\OPT\to4\lambda+3(1-\lambda)=3+\lambda$ for every
$\lambda$.  At $\lambda=2/3$ this limit is $11/3$, proving that the upper bound
is exact in every Euclidean dimension because the profiles embed isometrically
into every Euclidean space.
Finally, strategyproofness on Ptolemaic spaces follows from
\Cref{cor:mix-sp}.  The supremum is approached along the family and need not be
attained at a finite profile.
\end{proof}

Finally, \Cref{app:mixture-limits} shows that, for each mixing weight
$\lambda\in[0,1]$, the fixed mixture $\Mix{\lambda}$ has worst-case
approximation ratio is at least $\max\set{3+\lambda,4-(4\sqrt3-6)\cdot\lambda}$, so the best
ratio in this family lies between $(74+4\sqrt3)/23=3.517$ and $11/3$,
and its exact value remains open.
 
\section{A Lower Bound on the Approximation Ratio of Any Strategyproof Mechanism}
\label{sec:lowerbound}
We begin with a finite two-profile lower bound and then optimize its parameters
asymptotically.  The construction lies on the real line and uses two profiles
related by a single misreport.  Fix integers $2\le k<m$ and set $n=2m+k$.
Consider the two profiles
\begin{align}\label{eq:profiles}
  \lbprofA_{m,k}=\lranglefix{0^{m},(\tfrac12)^{k},1^{m}},
  \qquad
  \lbprofB_{m,k}=\lranglefix{-m,\,0^{m},(\tfrac12)^{k-1},1^{m}},
\end{align}
where $z^{\ell}$ denotes $\ell$ agents reporting $z$.  Thus $\lbprofB_{m,k}$ is
obtained from $\lbprofA_{m,k}$ when one designated agent of the midpoint block
misreports her location $\tfrac12$ as the distant point $-m$.  We suppress the
subscripts when they are clear.

The design is deliberate.  The two endpoint blocks of size $m$ make it costly
for a mechanism to ignore either end.  Because the midpoint agents are
co-located, they incur the same expected cost; thus the incentive constraint of
the one deviating agent controls a $\btheta{k}$ share of the social cost at both
profiles.  The far outlier at $-m$ forces the optimal solution to relocate a
facility, so the two profiles have very different optima.

\begin{theorem}\label{thm:lb-finite}
Let $2\le k<m$ and $n=2m+k$.  For every mechanism $\Mech$ that is
strategyproof on $\reals$, consider the two profiles $\lbprofA_{m,k}$ and $\lbprofB_{m,k}$ in \cref{eq:profiles}.  Then
\begin{align*}
  \max\set{\frac{\SC(\Mech,\lbprofA)}{\OPT(\lbprofA)},\;
           \frac{\SC(\Mech,\lbprofB)}{\OPT(\lbprofB)}}
  \;\ge\;
  \Rlb{m}{k}=\frac{m(2m-k-1)}{k(k-1)+2m(m-k)} .
\end{align*}
In particular no strategyproof mechanism for $n=2m+k$ agents has approximation
ratio smaller than $\Rlb{m}{k}$.
\end{theorem}

\begin{proof}
If the mechanism has infinite expected social cost at either profile the claim
is immediate, so assume both are finite; since $k\ge2$, all expectations below
are then finite as well.  For a profile $P$ write
$\dst_{P}(z)=\cost(\Mech(P),z)$ for the expected cost of a point $z$ under
the outcome law at $P$, and set $a=\dst_{\lbprofA}(\tfrac12)$ and
$b=\dst_{\lbprofB}(\tfrac12)$.
Note that $b$ evaluates the outcome law at $\lbprofB$ at the \emph{true}
location $\tfrac12$ of the designated agent, even though she reports $-m$ there.

Strategyproofness first provides a constraint connecting the two profiles.
The designated agent has true location $\tfrac12$.  Reporting truthfully
produces $\lbprofA$ and deviating to $-m$ produces $\lbprofB$, so
\Cref{def:sp} gives $a\le b$.

We next establish two geometric inequalities that hold for every deterministic
facility pair.
Fix any deterministic facility pair and abbreviate $\dst_{z}$ for the cost of
the point $z$ under it.  First,
\begin{align}\label{eq:g1}
  \dst_{0}+\dst_{1/2}+\dst_{1}\;\ge\;\tfrac12 .
\end{align}
Indeed, assign each of $0,\tfrac12,1$ to a nearest facility; two of them are
assigned to the same facility, and the sum of their distances to it is at least
their separation, which is at least $\tfrac12$.  Second,
\begin{align}\label{eq:g2}
  \dst_{-m}+m\,\dst_{0}+m\,\dst_{1}\;\ge\;m .
\end{align}
Assign $-m,0,1$ to nearest facilities; again two share a facility.  If $-m$ and
$0$ share it then $\dst_{-m}+\dst_{0}\ge m$; if $0$ and $1$ share it then
$m(\dst_{0}+\dst_{1})\ge m$; if $-m$ and $1$ share it then
$\dst_{-m}+\dst_{1}\ge m+1$.  As $m\ge1$ and all distances are nonnegative, each
case implies \cref{eq:g2}.

The inequalities in \cref{eq:g1,eq:g2} hold for every deterministic outcome and
therefore remain valid after taking expectations under any randomized outcome
law.

We now combine these geometric inequalities with the incentive constraint.
Let $t=\frac{k-1}{m-k}>0$, so that $t(m-k)=k-1$.  Apply \cref{eq:g1} at
$\lbprofA$, where the three occupied sites carry $m$, $k$, and $m$ agents.  At
$\lbprofB$, remove the contribution of the $k-1$ remaining midpoint agents and
apply \cref{eq:g2}.  These two applications give
\begin{align*}
  t\,\SC(\Mech,\lbprofA)+(k-1)a
  &=tm\bigl(\dst_{\lbprofA}(0)+\dst_{\lbprofA}(\tfrac12)
       +\dst_{\lbprofA}(1)\bigr)
  \ge\frac{tm}{2},\\
  \SC(\Mech,\lbprofB)-(k-1)b
  &=\dst_{\lbprofB}(-m)+m\,\dst_{\lbprofB}(0)+m\,\dst_{\lbprofB}(1)
  \ge m .
\end{align*}
Adding these inequalities and using $a\le b$ yields
\begin{align}\label{eq:combined}
  t\,\SC(\Mech,\lbprofA)+\SC(\Mech,\lbprofB)
  \;\ge\;m\Bigl(1+\frac t2\Bigr).
\end{align}

We next compute the optimal social cost of each profile.
At $\lbprofA$, since $m>k$, replacing the coefficient $m$ of each endpoint cost
by $k$ can only decrease the social cost.  Hence \cref{eq:g1} gives
$\SC(\lbprofA,\pair)\ge k/2$, while facilities at $\set{0,1}$ attain $k/2$.
Thus $\OPT(\lbprofA)=k/2$.  At $\lbprofB$, \cref{eq:g2} and the nonnegativity of
the midpoint costs give $\SC(\lbprofB,\pair)\ge m$, while facilities at
$\set{-m,\tfrac12}$ attain $m$.  Thus $\OPT(\lbprofB)=m$.

We can now derive the claimed ratio.
Suppose both profiles had ratio at most $\ratio$.  By the two optimum values,
$\SC(\Mech,\lbprofA)\le\ratio k/2$ and $\SC(\Mech,\lbprofB)\le\ratio m$.
Substituting into \cref{eq:combined} gives
$\ratio\bigl(m+\tfrac{tk}{2}\bigr)\ge m\bigl(1+\tfrac t2\bigr)$, and inserting
$t=(k-1)/(m-k)$ and simplifying gives
$\ratio\ge\frac{m(2m-k-1)}{k(k-1)+2m(m-k)}$.
\end{proof}

By optimizing the parameters $m$ and $k$, we obtain the following lower bound on the approximation ratio of any strategyproof mechanism for the
two-facility game.

\begin{theorem}\label{thm:lb-asymptotic}
Any strategyproof mechanism for the two-facility game has
approximation ratio at least $(1+\sqrt{2})/2$.
\end{theorem}

\begin{proof}
Let $k/m\to\alpha\in(0,1)$.  Then
$\Rlb{m}{k}\to\lbfun(\alpha)=\frac{2-\alpha}{\alpha^{2}-2\alpha+2}$, whose
derivative is
\begin{align*}
  \lbfun'(\alpha)=\frac{\alpha^{2}-4\alpha+2}{(\alpha^{2}-2\alpha+2)^{2}} .
\end{align*}
The unique root of $\alpha^{2}-4\alpha+2$ in $(0,1)$ is $\alpha^{*}=2-\sqrt2$.
Since $\lbfun(0)=\lbfun(1)=1$, this stationary point is the unique maximizer,
and $\lbfun(2-\sqrt2)=\lbconst$.  Taking $k_{m}=\lfloor(2-\sqrt2)m\rceil$, where
$\lfloor\cdot\rceil$ is rounding to the nearest integer, gives $2\le k_{m}<m$
for all large $m$ and $\Rlb{m}{k_{m}}\to\lbconst$.  By \Cref{thm:lb-finite},
for every $\varepsilon>0$ there are arbitrarily large $n$ at which every
strategyproof mechanism has an instance of ratio at least $\lbconst-\varepsilon$;
hence no mechanism attains a worst-case ratio below $\lbconst$.

\end{proof}
 
\section{Conclusion}
\label{sec:conclusion}
We introduced the Global Pair mechanism and proved that it is strategyproof on
every Ptolemaic metric space.  Although Global Pair and the Proportional
mechanism each have worst-case approximation ratio $4$, their profile-dependent
bounds vary in opposite directions with the dispersion parameter.  Randomizing
between them therefore yields a strategyproof mechanism with worst-case ratio
$11/3$.  We also established a lower bound of
$(1+\sqrt{2})/2\approx1.207$ for randomized strategyproof mechanisms, using a two-profile block-amplification construction.
Together, these results narrow the gap from
$[1.045,4]$ to $[(1+\sqrt{2})/2,11/3]$.  

Determining the optimal randomized strategyproof approximation ratio remains
open.
Since we have shown that any mixture of the Proportional and Global Pair mechanisms has worst-case ratio at least $3.517$, improving the upper bound below this value may require a new mechanism.  On the lower-bound side, our construction is limited to two profiles, and it is unclear whether more profiles can yield a better lower bound.  
Moreover, extending the Global Pair mechanism to more than two facilities in a strategyproof manner with bounded approximation ratio is an interesting open question.
We leave these questions for future work.
 
\section*{Acknowledgments}
\label{sec:ack}
The authors used ChatGPT Pro 5.6 to assist in formalizing certain proof steps, which were subsequently rewritten and independently verified by the authors. In particular, ChatGPT assisted with the analysis in \Cref{app:mixture-limits}, which bounds the optimal worst-case approximation ratio within the family of fixed mixtures over all mixing weights $\lambda\in[0,1]$. The authors take full responsibility for the correctness and presentation of all arguments in this paper.
 
\bibliographystyle{plainnat}

\appendix

\section{Truthfulness of the Global Pair mechanism}\label{app:gp-sp}

We prove \Cref{thm:gp-sp}.  Fix a Ptolemaic metric space, a strategic agent with
true location $x$ and reported location $x\primed$, and the other reports
$z_{1},\dots,z_{m}$.  Write
\begin{align*}
  \devd=\dst(x,x\primed),\qquad \tdist{i}=\dst(x,z_{i}),\qquad \rdist{i}=\dst(x\primed,z_{i}),
  \qquad \odist{i}{j}=\dst(z_{i},z_{j}),
\end{align*}
and $\Ssum=\sum_{i}\tdist{i}$, $\Tsum=\sum_{i}\rdist{i}$, $\pw=\sum_{i<j}\odist{i}{j}$.  Let
\begin{align*}
  \Aterm=\sum_{i<j}\odist{i}{j}\min\set{\tdist{i},\tdist{j}},
  \qquad
  \Bterm=\sum_{i}\rdist{i}\min\set{\devd,\tdist{i}} .
\end{align*}
By \Cref{def:gp}, the agent's expected cost when truthful is
$C_{x}=\Aterm/(\pw+\Ssum)$, and when she reports $x\primed$ it is
$C_{x\primed}=(\Aterm+\Bterm)/(\pw+\Tsum)$: the
denominators are the total pair weights of the two profiles, the numerator $\Aterm$
collects the pairs among the other agents, and $\Bterm$ collects the pairs involving
the strategic agent, whose cost to a truthful agent at $x$ is
$\min\set{\devd,\tdist{i}}$.  Hence $C_{x}\le C_{x\primed}$ is equivalent to
\begin{align}\label{eq:sp-target}
  \Aterm(\Tsum-\Ssum)\le \Bterm(\pw+\Ssum).
\end{align}
Since $\Ssum\ge0$ it suffices to prove the stronger statement
\begin{align}\label{eq:sp-strong}
  \Aterm(\Tsum-\Ssum)\le \Bterm\pw .
\end{align}

Set $\delta_{i}=\rdist{i}-\tdist{i}$ and $\eta_{i}=\rdist{i}\min\set{\devd,\tdist{i}}$.  Expanding
both sides of \cref{eq:sp-strong} and grouping by the pair $\set{j,k}$
generating each $\odist{j}{k}$ gives the identity
\begin{align}\label{eq:sp-expand}
  \Bterm\pw-\Aterm(\Tsum-\Ssum)=\sum_{j<k}\odist{j}{k}\sum_{i}\bigl[\eta_{i}-\min\set{\tdist{j},\tdist{k}}\delta_{i}\bigr].
\end{align}

\begin{claim}\label{cl:eta}
For every $i$ and every $0\le t\le \tdist{i}$ we have $\eta_{i}-t\delta_{i}\ge0$.
\end{claim}

\begin{proof}
If $\delta_{i}\le0$ this is immediate.  If $\delta_{i}>0$ the left-hand side is
decreasing in $t$, so it is minimized at $t=\tdist{i}$, where it equals
$\rdist{i}\min\set{\devd,\tdist{i}}-\tdist{i}(\rdist{i}-\tdist{i})$.  When $\tdist{i}\le r$ this is
$\tdist{i}^{2}\ge0$.  When $\tdist{i}\ge r$ it equals $\rdist{i}r-\tdist{i}\rdist{i}+\tdist{i}^{2}$,
which is at least $r^{2}$ because $\rdist{i}\le \tdist{i}+r$ by the triangle
inequality.
\end{proof}

By \Cref{cl:eta} the terms of \cref{eq:sp-expand} with $i\in\set{j,k}$ are
nonnegative, since there $\min\set{\tdist{j},\tdist{k}}\le \tdist{i}$.  It therefore remains
to handle the terms with $i\notin\set{j,k}$, which we group by unordered triples.
Relabel a triple so that $\tdist{1}\le \tdist{2}\le \tdist{3}$; its contribution is
\begin{align*}
  \Psi=\odist{2}{3}(\eta_{1}-\tdist{2}\delta_{1})+\odist{1}{3}(\eta_{2}-\tdist{1}\delta_{2})
       +\odist{1}{2}(\eta_{3}-\tdist{1}\delta_{3}).
\end{align*}
By \Cref{cl:eta} the last two brackets are nonnegative, so writing
\begin{align*}
  N=\tdist{2}\delta_{1}-\eta_{1},\qquad
  C_{2}=\eta_{2}-\tdist{1}\delta_{2},\qquad
  C_{3}=\eta_{3}-\tdist{1}\delta_{3},
\end{align*}
only $N$ can be positive, and it suffices to prove
\begin{align}\label{eq:triple-target}
  N \odist{2}{3}\le C_{2}\odist{1}{3}+C_{3}\odist{1}{2}.
\end{align}
We distinguish three cases.

\xhdr{Case 1: $r\ge \tdist{2}$.}
Then $\min\set{\devd,\tdist{j}}=\tdist{j}$ for $j\le2$, so
$N=\tdist{2}(\rdist{1}-\tdist{1})-\rdist{1}\tdist{1}\le \rdist{1}(\tdist{2}-\tdist{1})$ and, for $j\in\set{2,3}$,
$C_{j}\ge \rdist{j}(\tdist{2}-\tdist{1})$.  The Ptolemy inequality in \cref{eq:ptolemy}, applied to the four points
$x\primed,z_{1},z_{2},z_{3}$, reads $\rdist{1}\odist{2}{3}\le \rdist{2}\odist{1}{3}+\rdist{3}\odist{1}{2}$,
which gives \cref{eq:triple-target}.

\xhdr{Case 2: $r<\tdist{2}$ and $\tdist{1}\ge r$.}
Then $\min\set{\devd,\tdist{i}}=r$ for all $i$ in the triple, so
$N=\tdist{2}\delta_{1}-\rdist{1}r\le K\deq r(\tdist{2}-\tdist{1}-r)$, using $\rdist{1}\le \tdist{1}+r$.
If $K\le0$ there is nothing to prove.  Otherwise $C_{2},C_{3}\ge K$, and the
triangle inequality $\odist{2}{3}\le \odist{1}{3}+\odist{1}{2}$ gives \cref{eq:triple-target}.

\xhdr{Case 3: $a\deq \tdist{1}<r<b\deq \tdist{2}\le c\deq \tdist{3}$.}
Write $\pi=\rdist{2}$ and $\theta=\rdist{3}$.  Here $\eta_{1}=\rdist{1}\tdist{1}$ while
$\eta_{2}=\pi r$ and $\eta_{3}=\theta r$, so
\begin{align*}
  N\le K\deq r(b-a)-a^{2},\qquad
  C_{2}=(r-a)\pi+ab,\qquad
  C_{3}=(r-a)\theta+ac .
\end{align*}
If $K\le0$ the claim is immediate, so assume $K>0$.  The triangle inequalities
give $\odist{1}{3}\ge c-a$, $\odist{1}{2}\ge b-a$ and
$\odist{2}{3}\le\min\set{b+c,\pi+\theta}$, so it suffices to prove
\begin{align}\label{eq:case3}
  L\deq C_{2}(c-a)+C_{3}(b-a)\;\ge\;K\min\set{b+c,\pi+\theta}.
\end{align}
Suppose first $\pi+\theta\le b+c$.  Expanding,
\begin{align*}
  L-K(\pi+\theta)=(\pi+\theta)a(2a-b)+\pi(r-a)(c-b)+a\bigl[2bc-a(b+c)\bigr],
\end{align*}
which is visibly nonnegative when $b\le2a$.  When $b>2a$, using
$\pi+\theta\le b+c$ and $\pi\ge b-r$, the expression is at least
$(b-r)(r-a)(c-b)+a\bigl[b(c-b)+a(b+c)\bigr]\ge0$.  Suppose instead
$\pi+\theta\ge b+c$.  Minimizing the positive weighted $\pi,\theta$ terms subject
to $\pi\ge b-r$, $\theta\ge c-r$, $\pi+\theta\ge b+c$ and the triangle upper
bounds $\pi\le b+r$, $\theta\le c+r$, the minimum occurs at $\pi=b-r$ and
$\theta=c+r$, where
\begin{align*}
  L-K(b+c)\ge r(c-b)(a+b-r)+a^{2}(b+c)\ge0 .
\end{align*}
This proves \cref{eq:case3} and hence \cref{eq:triple-target}.

Every grouped contribution in \cref{eq:sp-expand} is therefore nonnegative,
which establishes \cref{eq:sp-strong}, hence \cref{eq:sp-target}, and completes
the proof of \Cref{thm:gp-sp}. \qed

\section{Failure of truthfulness on a general metric}\label{app:gp-counterexample}

Let $\dom$ be the vertex set of the complete bipartite graph $K_{2,3}$, with
bipartition $\set{x,y}$ and $\set{z_{1},z_{2},z_{3}}$, equipped with the
unit-edge shortest-path metric.  Thus
\begin{align*}
  \dst(x,z_{i})=\dst(y,z_{i})=1,
  \qquad
  \dst(x,y)=\dst(z_{i},z_{j})=2\qquad(i\ne j).
\end{align*}
This metric is not Ptolemaic: for the four points $x,y,z_{1},z_{2}$,
\begin{align*}
  \dst(x,y)\cdot\dst(z_{1},z_{2})=4\;>\;2
  =\dst(x,z_{1})\cdot\dst(y,z_{2})+\dst(x,z_{2})\cdot\dst(y,z_{1}),
\end{align*}
so \Cref{thm:gp-sp} does not apply to it.  Take $n=14$ agents: eleven at $x$,
one of whom is strategic, and one at each of $z_{1},z_{2},z_{3}$.

\xhdr{Truthful reporting.}
The reported profile is eleven agents at $x$ and one at each $z_{i}$.  Pairs of
agents at $x$ carry weight $0$, the $33$ pairs of the form $\set{x,z_{i}}$ carry
weight $1$ each, and the $3$ pairs $\set{z_{i},z_{j}}$ carry weight $2$ each, so $\pw=11\cdot3\cdot1+\binom{3}{2}\cdot2=39$.
Whenever a selected pair contains an agent reporting $x$, a facility is opened
at $x$ and the strategic agent pays nothing; she pays $1$ exactly on the pairs
$\set{z_{i},z_{j}}$, of total weight $6$.  Her expected cost is therefore $C_{x}=\frac{6}{39}=\frac{2}{13}$.

\xhdr{Misreporting $x\to y$.}
The reported profile becomes ten agents at $x$, one at $y$, and one at each
$z_{i}$.  The strategic agent is still located at $x$, so her cost on a selected
pair is her distance to the nearer of the two opened facilities:
\begin{center}
\begin{tabular}{lcc}
  \toprule
  Pair type & Total weight & True cost to the deviator\\
  \midrule
  $\set{x,z_{i}}$   & $10\cdot3\cdot1=30$ & $0$\\
  $\set{x,y}$       & $10\cdot1\cdot2=20$ & $0$\\
  $\set{y,z_{i}}$   & $1\cdot3\cdot1=3$   & $1$\\
  $\set{z_{i},z_{j}}$ & $3\cdot2=6$       & $1$\\
  \bottomrule
\end{tabular}
\end{center}
Hence the total pair weight is $\pw'=30+20+3+6=59$ and
$C_{y}=\frac{3+6}{59}=\frac{9}{59}$.

Since $C_{x}-C_{y}=\frac{1}{767}>0$,
the misreport strictly lowers the strategic agent's expected cost, so the Global
Pair mechanism is not strategyproof on $\dom$.

The block of co-located agents cannot be shrunk.  Running the same computation
with $b$ agents at $x$ in place of eleven gives $\pw=3b+6$ and $\pw'=5b+4$,
hence
\begin{align*}
  C_{x}=\frac{2}{b+2},
  \qquad
  C_{y}=\frac{9}{5b+4},
  \qquad
  C_{x}-C_{y}=\frac{b-10}{(b+2)(5b+4)} ,
\end{align*}
which is negative for $b\le9$ and zero for $b=10$.  Thus $b=11$, that is
$n=14$, is the smallest instance of this family at which the deviation is
profitable.

\section{The metric triple lemma}\label{app:triple}

\begin{lemma}[Metric triple lemma]\label{lem:triple}
Let $\sep=1$ and $H(c)=3c+c\min\set{c,1}$.  For every triple of agents
$\set{i,j,k}$ in every metric space,
\begin{align*}
  \Phi_{ijk}\le\dst(\loc{j},\loc{k})\cdot H(\ocost{i})
   +\dst(\loc{i},\loc{k})\cdot H(\ocost{j})
   +\dst(\loc{i},\loc{j})\cdot H(\ocost{k}),
\end{align*}
with $\Phi_{ijk}$ as in \cref{eq:triple-weight}.
\end{lemma}

\begin{proof}
Scale so that $\sep=1$ and write
$H(c)=3c+c\min\set{c,1}$.

\xhdr{All three agents in one cluster.}
Label the edges $\dst_{12}=\alpha$, $\dst_{13}=\beta$, $\dst_{23}=\gamma$ and
the optimal costs $A,B,C$.  Since $\alpha\le A+B$ and $\gamma\le B+C$,
\begin{align*}
  \Phi\le2\beta(A+B)+\alpha(B+C)\le3\gamma A+3\beta B+3\alpha C ,
\end{align*}
so the linear part $3c$ of $H$ already suffices.

\xhdr{Exactly two agents in one cluster.}
Let the two co-clustered agents $\{i,j\}$ have optimal costs $A,B$ and mutual distance
$s=d(x_i,x_j)$, and let the singleton $\{k\}$ have cost $C$.  Denote the two cross edges by $d(x_i,x_k) =u\le v = d(x_j,x_k)$,
where $u$ is incident to the agent of cost $A$.  Then
\begin{align}\label{eq:split-tri}
  s\le A+B,\qquad v-u\le s,\qquad u\ge1-A-C .
\end{align}
The right-hand side of \Cref{lem:triple} equals
$v\cdot H(A)+u\cdot H(B)+s\cdot H(C)$.  Write
$\Delta=v-u$ and consider the three possible edge orders.

\emph{(i) $s\le u\le v$.}  Here $\Phi=s(2u+v)$.  Using $s\le A+B$, $3vA+3uB+3sC - \Phi =3u(A+B-s)+\Delta (3A-s)+3sC \ge 3sC-(B-2A)\Delta$, which is
nonnegative if $B\le2A$.  Otherwise $\Delta\le s$, so it suffices that
$uB\min\set{B,1}+3sC\ge s(B-2A)$.  If $B\ge1$ this follows from $u\ge s$.  If
$B<1$ and $3C\ge B-2A$ it is immediate.  In the remaining subcase of $B<1$ and $3C< B-2A$,  
Using \cref{eq:split-tri}, we obtain the lower bound
\begin{align*}
    uB\min\set{B,1}+ s(3C-B+2A) \ge & 
  (1-A-C)B^{2}+(A+B)(3C-B+2A)\\
  = & A(B-B^{2}+2A)+C(3A+3B-B^{2})\ge0 .
\end{align*}

\emph{(ii) $u\le s\le v$.}  Here $\Phi=2us+uv\le u\bigl(2(A+B)+v\bigr)$, and $3vA+3uB+3sC - \Phi \ge u(A+B-u)+\Delta(3A-u)+3sC$, nonnegative when $u\le3A$.
Otherwise, using $\Delta\le s$ and $s\le A+B$ and retaining the extra term
$uB\min\set{B,1}$, the residual is at least $uB\min\set{B,1}+3s(A+C)-u^{2}$.
For $B\ge1$ this is at least $u(B+3A-u) \ge0$; for $B<1$ it is at least
$u\bigl[B^{2}+3(A+C)-u\bigr]\ge u(2A+3C-B+B^{2}) \ge u(2A+3C-1+B) \ge u(2A+3C-2A-C) \ge0$, the last step using $u\le A+B$ together with
$1-B\le2A+C$ and $B<1$, both consequences of \cref{eq:split-tri}.

\emph{(iii) $u\le v\le s$.}  Here $\Phi=2uv+us\le2uv+u(A+B)$ and $3vA+3uB+3sC - \Phi \ge 2u(A+B-u)+\Delta(3A-2u)+3sC$, nonnegative when $3A\ge2u$.
Otherwise $\Delta\le s-u$ and the expression is at least $2u(A+B-u)+(s-u)(3A-2u)+3sC \ge 3A(s-u)+3sC\ge0$.
\end{proof}

\section{Proof of \Cref{thm:prop-interp}}\label{app:prop-interp}

We prove \Cref{thm:prop-interp} for $\OPT(\prof)>0$; the zero-optimum case is
handled in the main text.  Retain the notation of
\Cref{def:dispersion}: clusters $\clus{1},\clus{2}$ with
$\nclus{1}=\setsize{\clus{1}}$, $\nclus{2}=\setsize{\clus{2}}$, centers
$\ofac{1},\ofac{2}$, separation $\sep$, cluster optima $\Rsum_{1},\Rsum_{2}$, and for
$i\in\clus{1}$ the quantities $U_{i}=\sum_{k\in\clus{1}}\dst(\loc{i},\loc{k})$ and
$V_{i}=\sum_{k\in\clus{2}}\dst(\loc{i},\loc{k})$.

\subsection{A fixed-anchor lemma}

\begin{lemma}\label{lem:fixed-anchor}
Let $C_{i}$ be the expected social cost of the Proportional mechanism
conditional on anchoring at $i\in\clus{1}$.  Then
\begin{align}\label{eq:fixed-anchor}
  C_{i}\le 2U_{i}+3\Rsum_{2}-\frac{U_{i}(U_{i}+3\Rsum_{2})}{U_{i}+V_{i}} .
\end{align}
\end{lemma}

\begin{proof}
For $j,k\in\clus{2}$ write $v_{j}=\dst(\loc{i},\loc{j})$, $r_{j}=\dst(\ofac{2},\loc{j})$ and
$\dst_{jk}=\dst(\loc{j},\loc{k})$.  We claim the ordered-pair inequality
\begin{align}\label{eq:ordered-pair}
  v_{j}\min\set{v_{k},\dst_{jk}}\le v_{j}r_{k}+2r_{j}v_{k}.
\end{align}
If $v_{j}\le2v_{k}$ then
$v_{j}\min\set{v_{k},\dst_{jk}}\le v_{j}\dst_{jk}\le v_{j}(r_{j}+r_{k})\le2r_{j}v_{k}+v_{j}r_{k}$.
If $v_{j}>2v_{k}$ then $\dst_{jk}\ge v_{j}-v_{k}>v_{k}$, so the left-hand side is
$v_{j}v_{k}$, and $r_{j}+r_{k}\ge\dst_{jk}\ge v_{j}-v_{k}$ again gives
\cref{eq:ordered-pair}.  Summing over ordered pairs,
\begin{align}\label{eq:cluster-B}
  \sum_{j\in\clus{2}}v_{j}\sum_{k\in\clus{2}}\min\set{v_{k},\dst_{jk}}\le3\Rsum_{2}V_{i}.
\end{align}
Now condition on the second draw.  If it lands in $A$, serving all agents from
the anchor costs at most $U_{i}+V_{i}$.  If it lands in $B$, cluster $\clus{1}$ costs at
most $U_{i}$ while \cref{eq:cluster-B} controls cluster $\clus{2}$.  Since the second
draw selects $j$ with probability proportional to $\dst(\loc{i},\loc{j})$, the
unnormalized conditional numerator is at most
$U_{i}(U_{i}+V_{i})+U_{i}V_{i}+3\Rsum_{2}V_{i}$; dividing by $U_{i}+V_{i}$ gives
\cref{eq:fixed-anchor}.  The case $U_{i}+V_{i}=0$ is trivial.
\end{proof}

\subsection{Retaining the slack}

Set $\Lambda_{1}=\sum_{i\in\clus{1}}U_{i}(U_{i}+3\Rsum_{2})/(U_{i}+V_{i})$ and define
$\Lambda_{2}$ symmetrically.  Since
$\sum_{i\in\clus{1}}U_{i}\le2\nclus{1}\Rsum_{1}$ and
$\sum_{i\in\clus{2}}U_{i}\le2\nclus{2}\Rsum_{2}$, averaging
\cref{eq:fixed-anchor} over the
uniform anchor gives
\begin{align}\label{eq:prop-avg}
  \Prop\le3\Rsum+\frac{\nclus{1}\Rsum_{1}+\nclus{2}\Rsum_{2}}{n}
    -\frac{\Lambda_{1}+\Lambda_{2}}{n}.
\end{align}
Discarding the last term recovers the factor $4$; we bound it from below
instead.  Let
$\Qsum_{1}=\sum_{i\in\clus{1}}\ocost{i}\min\set{\ocost{i},\sep}$, define
$\Qsum_{2}$ symmetrically, and set $\Qsum=\Qsum_{1}+\Qsum_{2}$.  The claim is
\cref{eq:lambda-lb}.  We prove the first; the second is symmetric.

We first record that each optimal facility is a \emph{$1$-median} of its own
cluster: $\ofac{1}$ minimizes $\sum_{k\in\clus{1}}\dst(\loc{k},z)$ over
$z\in\dom$, since otherwise moving it to a better point for $\clus{1}$, with
$\ofac{2}$ left in place, would lower the social cost and contradict the
optimality of $(\ofac{1},\ofac{2})$; symmetrically for $\ofac{2}$.  In
particular
\begin{align}\label{eq:one-median}
  \sum_{k\in\clus{1}}\dst(\loc{k},z)\;\ge\;\Rsum_{1}
  \qquad\text{for every }z\in\dom .
\end{align}

For $i\in\clus{1}$ this gives
\begin{align}\label{eq:UV-bounds}
  V_{i}\le \nclus{2}(\sep+\ocost{i})+\Rsum_{2},
  \qquad
  U_{i}\ge\Rsum_{1},
  \qquad
  U_{i}\ge \nclus{1}\ocost{i}-\Rsum_{1}.
\end{align}
The middle bound is \cref{eq:one-median} at $z=\loc{i}$.  The other two are
triangle inequalities: $\dst(\loc{i},\loc{k})\le\ocost{i}+\sep+\ocost{k}$ for
$k\in\clus{2}$, summed over the $\nclus{2}$ agents of $\clus{2}$, and
$\dst(\loc{i},\loc{k})\ge\ocost{i}-\ocost{k}$ for $k\in\clus{1}$, summed over
the $\nclus{1}$ agents of $\clus{1}$.
Putting $Q_{i}=\nclus{2}(\sep+\ocost{i})$ and writing
$\lambda_{i}^{1}=U_{i}(U_{i}+3\Rsum_{2})/(U_{i}+V_{i})$,
\begin{align}\label{eq:lambda-i}
  \lambda_{i}^{1}\ge\frac{U_{i}(U_{i}+3\Rsum_{2})}{U_{i}+Q_{i}+\Rsum_{2}}
  \ge\frac{U_{i}^{2}}{U_{i}+Q_{i}},
\end{align}
the second step because
$(U_{i}+3\Rsum_{2})(U_{i}+Q_{i})-U_{i}(U_{i}+Q_{i}+\Rsum_{2})=2U_{i}\Rsum_{2}+3Q_{i}\Rsum_{2}\ge0$.

If $\nclus{2}\ge \nclus{1}$ then $\Qsum_{1}\le\sep\Rsum_{1}$ and $\nclus{2}\ge n/2$, so
$\Lambda_{1}+\nclus{2}\Rsum_{1}\ge \nclus{2}\Rsum_{1}
\ge\frac n2\Rsum_{1}\ge\frac{n\Qsum_{1}}{2\sep}$
and we are done.  Assume therefore $\nclus{1}\ge \nclus{2}$.

\subsection{The scalar certificate for
  \texorpdfstring{$\nclus{1}\ge \nclus{2}$}{n1 >= n2}}

Scale $\sep$ to $1$ and set
\begin{align*}
  r=\Rsum_{1},\qquad z=\ocost{i},\qquad
  m=\max\set{r,\nclus{1}z-r},\qquad s=\min\set{z,1},
\end{align*}
so that \cref{eq:UV-bounds,eq:lambda-i} give
$\lambda_{i}^{1}\ge m^{2}/(m+\nclus{2}(1+z))$.
It suffices to prove the scalar inequality
\begin{align}\label{eq:scalar}
  \frac{m^{2}}{m+\nclus{2}(1+z)}+\nclus{2}z-\frac{\nclus{1}+\nclus{2}}{2}zs
  \;\ge\;\frac{\nclus{1}r}{r+\nclus{2}}\left(\frac r{\nclus{1}}-z\right),
\end{align}
since summing \cref{eq:scalar} over $i\in\clus{1}$ makes the right-hand side vanish,
because $\sum_{i\in\clus{1}}(\Rsum_{1}/\nclus{1}-\ocost{i})=0$, and restoring $\sep$ then
yields the first inequality of \cref{eq:lambda-lb}.

Let $\Delta=\nclus{1}-\nclus{2}\ge0$.  Suppose first $m=r$, so that
$r\ge \nclus{1}z/2$.  Subtracting
the right-hand side of \cref{eq:scalar} factors the residual as
$z\bigl[T-\tfrac{\Delta s}{2}+\nclus{2}(1-s)\bigr]$ with
\begin{align*}
  T=\frac{r\bigl[\Delta r+\nclus{1}\nclus{2}(1+z)\bigr]}
    {(r+\nclus{2})\bigl(r+\nclus{2}(1+z)\bigr)} .
\end{align*}
For $0\le z\le1$ it suffices that $T\ge\Delta z/2$; clearing denominators, the
difference is
$\Delta(2-z)r^{2}
+\nclus{2}r\bigl[2\nclus{1}(1+z)-\Delta z(2+z)\bigr]
-\Delta z\nclus{2}^{2}(1+z)$,
whose bracket is at least $2\nclus{2}(1+z)$, and
$r\ge \nclus{1}z/2\ge\Delta z/2$, so the whole
expression is nonnegative.  For $z\ge1$ it suffices that $T\ge\Delta/2$; the
cleared difference is
$\Delta r^{2}+\nclus{2}r\bigl(2\nclus{2}+(\nclus{1}+\nclus{2})z\bigr)
-\Delta \nclus{2}^{2}(1+z)$,
and using $r\ge \nclus{1}z/2$ it suffices that
$\nclus{1}z\bigl(2\nclus{2}+(\nclus{1}+\nclus{2})z\bigr)
-2\Delta \nclus{2}(1+z)\ge0$, whose left-hand side is increasing
for $z\ge1$ and equals $\nclus{1}^{2}-\nclus{1}\nclus{2}+4\nclus{2}^{2}>0$ at $z=1$.

Suppose instead $m=\nclus{1}z-r$, so $m\ge r$ and $m+r=\nclus{1}z$.  Let
$\varkappa=\Delta s/(2\nclus{1})$.  From
$2\nclus{1}-\Delta s\ge2\nclus{1}-\Delta\ge \nclus{1}+\nclus{2}$,
$m\ge \nclus{1}z/2$, and
$\nclus{1}(\nclus{1}+\nclus{2})\ge4\nclus{2}\Delta$, the last because
$\nclus{1}^{2}-3\nclus{1}\nclus{2}+4\nclus{2}^{2}>0$, we have
$m(2\nclus{1}-\Delta s)\ge \nclus{1}z/2\cdot(\nclus{1}+\nclus{2})
\ge2\nclus{2}\Delta z\ge\Delta s\nclus{2}(1+z)$.
This inequality gives
$m/(m+\nclus{2}(1+z))\ge\varkappa$; the same estimates with
$r\le m$ give
$m/(r+\nclus{2})\ge m/(m+\nclus{2}(1+z))\ge\varkappa$.  Therefore
\begin{align*}
  \frac{m^{2}}{m+\nclus{2}(1+z)}+\frac{rm}{r+\nclus{2}}
  \ge\varkappa(m+r)=\frac{\Delta zs}{2},
\end{align*}
which proves \cref{eq:scalar} because the remaining term $\nclus{2}z(1-s)$ is
nonnegative.

\subsection{Completing the proof}

Adding the two inequalities of \cref{eq:lambda-lb} gives
$\Lambda_{1}+\Lambda_{2}+\nclus{2}\Rsum_{1}+\nclus{1}\Rsum_{2}
\ge n\Qsum/(2\sep)$.  Substituting
into \cref{eq:prop-avg},
\begin{align*}
  \Prop\le3\Rsum+\frac{\nclus{1}\Rsum_{1}+\nclus{2}\Rsum_{2}}{n}
    -\frac{\Lambda_{1}+\Lambda_{2}}{n}
  \le4\Rsum-\frac{\Qsum}{2\sep}
  =\Bigl(4-\frac{\dispersion}{2}\Bigr)\Rsum . \qedhere
\end{align*}
\qed

\subsection{The exact approximation ratio of the Global Pair mechanism}
\label{app:gp-exact}

We establish the matching upper and lower bounds
\begin{align*}
  \ratio_{\dimension}(\GPair)=4
  \qquad\text{for every }\dimension\ge1.
\end{align*}

\begin{proof}
For the upper bound, the interpolation inequality in \Cref{thm:gp-interp},
together with $\dispersion\le1$, gives
\begin{align*}
  \SC(\GPair,\prof)
  \le(3+\dispersion)\OPT(\prof)
  \le4\OPT(\prof).
\end{align*}
Every finite Euclidean profile admits an optimal two-facility solution, so
$\ratio_{\dimension}(\GPair)\le4$.

For the lower bound, let $m\ge1$ and consider the profile
\begin{align*}
  \prof^{m}=\lranglefix{0^{m},1^{m},2},
\end{align*}
with $m$ agents at $0$, $m$ agents at $1$, and one agent at $2$.  This profile
embeds isometrically into $\reals^{\dimension}$ for every $\dimension\ge1$.

We first show that $\OPT(\prof^{m})=1$.  Placing facilities at $0$ and $1$
yields social cost $1$, because only the agent at $2$ incurs positive cost.
For the reverse inequality, consider any two facility locations and assign
each occupied location to a nearest facility.  By the pigeonhole principle,
two occupied locations, say $x$ and $y$, are assigned to the same facility
$\ell$.  If their multiplicities are $a,b\ge1$, then their total contribution
is at least
\begin{align*}
  a\dst(x,\ell)+b\dst(y,\ell)
  &\ge\min\set{a,b}
    \bigl(\dst(x,\ell)+\dst(y,\ell)\bigr) \\
  &\ge\min\set{a,b}\dst(x,y)
   \ge1.
\end{align*}
Thus every two-facility solution has social cost at least $1$, even if the
facilities lie off the line in a higher-dimensional space.

We next compute the expected social cost of the Global Pair mechanism.
Co-located pairs have weight zero, and the remaining unordered pairs fall into
three types:
\begin{itemize}
  \item There are $m^{2}$ pairs between $0$ and $1$.  Each has weight $1$ and
        produces social cost $1$.
  \item There are $m$ pairs between $0$ and $2$.  Each has weight $2$ and
        produces social cost $m$.
  \item There are $m$ pairs between $1$ and $2$.  Each has weight $1$ and
        produces social cost $m$.
\end{itemize}
Hence the total pair weight is
\begin{align*}
  \pw_m=m^{2}+2m+m=m^{2}+3m,
\end{align*}
and the corresponding weighted sum of social costs---the numerator of the
expectation---is
\begin{align*}
  \pw_m\SC(\GPair,\prof^{m})
  &=m^{2}\cdot1\cdot1
    +m\cdot2\cdot m
    +m\cdot1\cdot m \\
  &=4m^{2}.
\end{align*}
Therefore, using $\OPT(\prof^{m})=1$,
\begin{align*}
  \frac{\SC(\GPair,\prof^{m})}{\OPT(\prof^{m})}
  =\frac{4m^{2}}{m^{2}+3m}
  =\frac{4m}{m+3}
  \longrightarrow4.
\end{align*}
Taking the supremum over profiles gives $\ratio_{\dimension}(\GPair)\ge4$.
Together with the upper bound, this proves
$\ratio_{\dimension}(\GPair)=4$.
\end{proof}

\subsection{Limits of the mixture family}
\label{app:mixture-limits}

\Cref{thm:mixture-exact} identifies the exact ratio at the mixing weight
$\lambda=2/3$, but it does not determine the best mixing weight over the whole
family.  We now give a family-wide lower bound using one concentrated-profile
family and one spread-out-profile family.  Define the sharp coefficient
$\kcon$, the balancing weight $\lamstar$, and the balanced frontier value
$\rhostar$ by
\begin{align*}
  \kcon
    &\deq4\sqrt3-6=0.9282\ldots, \\
  \lamstar
    &\deq\frac{1}{1+\kcon}
     =\frac{1}{4\sqrt3-5}
     =\frac{5+4\sqrt3}{23}
     =0.5186\ldots, \\
  \rhostar
    &\deq3+\lamstar=\frac{74+4\sqrt3}{23}=3.5186\ldots.
\end{align*}
Define the lower-frontier function $L\colon[0,1]\to\reals$ by
\begin{align*}
  L(\lambda)\deq\max\set{3+\lambda,\;4-\kcon\cdot\lambda}.
\end{align*}

\begin{proposition}[Lower frontier]\label{prop:frontier}
For every mixing weight $\lambda\in[0,1]$, the worst-case approximation ratio
$\ratio(\Mix{\lambda})$ of the fixed mixture $\Mix{\lambda}$ satisfies
\begin{align*}
  \ratio(\Mix{\lambda})\ge L(\lambda).
\end{align*}
Consequently, the optimal worst-case ratio over all fixed mixing weights satisfies
\begin{align*}
  \inf_{\lambda\in[0,1]}\ratio(\Mix{\lambda})\ge\rhostar.
\end{align*}
\end{proposition}

\begin{proof}
Fix a mixing weight $\lambda\in[0,1]$.

\xhdr{The concentrated-profile branch.}
For every integer $m\ge2$, define the perturbation
$\varepsilon_m\deq m^{-2}$ and define the concentrated profile
$\prof_{\mathrm{c}}^{m}$ to have $m$ agents at $0$, one agent at
$\varepsilon_m$, and one agent at $1$; that is,
\begin{align*}
  \prof_{\mathrm{c}}^{m}
  \deq\lranglefix{0^{m},\varepsilon_m,1}.
\end{align*}
The optimal social cost of the concentrated profile is
$\OPT(\prof_{\mathrm{c}}^{m})=\varepsilon_m$.  Indeed, facilities at $0$ and
$1$ attain cost $\varepsilon_m$; conversely, two of the three occupied
locations must be assigned to the same facility, and their contribution is at
least their separation by the triangle inequality, which is at least
$\varepsilon_m$.  Conditioning on the Proportional anchor and enumerating the
Global Pair selections gives the normalized mechanism costs
\begin{align*}
  \frac{\SC(\Prop,\prof_{\mathrm{c}}^{m})}
       {\OPT(\prof_{\mathrm{c}}^{m})}
  &=\frac{m}{m+2}\cdot\left[
      \frac{2-\varepsilon_m}{1+\varepsilon_m}
      +\frac{2\cdot(1-\varepsilon_m)}{m\cdot\varepsilon_m+1-\varepsilon_m}
      +\frac{2-\varepsilon_m}{m+1-\varepsilon_m}
    \right]
    \longrightarrow4, \\
  \frac{\SC(\GPair,\prof_{\mathrm{c}}^{m})}
       {\OPT(\prof_{\mathrm{c}}^{m})}
  &=\frac{m\cdot(3-2\varepsilon_m)}
          {m\cdot(1+\varepsilon_m)+1-\varepsilon_m}
    \longrightarrow3.
\end{align*}
Therefore, the normalized cost of the fixed mixture on the concentrated
profile satisfies
\begin{align*}
  \frac{\SC(\Mix{\lambda},\prof_{\mathrm{c}}^{m})}
       {\OPT(\prof_{\mathrm{c}}^{m})}
  &\longrightarrow
    4\cdot\lambda+3\cdot(1-\lambda)
   =3+\lambda.
\end{align*}
Taking the supremum over profiles proves the concentrated-profile lower bound
\begin{align}\label{eq:frontier-concentrated}
  \ratio(\Mix{\lambda})\ge3+\lambda.
\end{align}

\xhdr{The spread-out-profile branch.}
For positive integer masses $A$ and $B$, define the spread-out profile
$\prof_{\mathrm{s}}^{A,B}$ to have mass $A$ at $0$, mass $B$ at $1$, and one
agent at $2$; that is,
\begin{align*}
  \prof_{\mathrm{s}}^{A,B}
  \deq\lranglefix{0^{A},1^{B},2}.
\end{align*}
The optimal social cost of the spread-out profile is
$\OPT(\prof_{\mathrm{s}}^{A,B})=1$.  Indeed, facilities at $0$ and $1$ attain
cost $1$.  Conversely, assign each of the three occupied locations to a nearer
facility.  Two occupied locations must be assigned to the same facility; since
each has mass at least $1$ and their separation is at least $1$, the triangle
inequality shows that their combined cost is at least $1$.  Enumerating the
weighted pairs and the Proportional anchors gives the expected mechanism costs
on the spread-out profile:
\begin{align*}
  \SC(\GPair,\prof_{\mathrm{s}}^{A,B})
  &=\frac{4A\cdot B}{A\cdot B+2A+B}, \\
  \SC(\Prop,\prof_{\mathrm{s}}^{A,B})
  &=\frac{1}{A+B+1}\cdot\left[
      \frac{3A\cdot B}{B+2}
      +\frac{2A\cdot B}{A+1}
      +\frac{3A\cdot B}{2A+B}
    \right].
\end{align*}

Fix a mass-ratio parameter $t>0$ and let the positive integer masses $A$ and
$B$ tend to infinity subject to $A/B\to t$.  Define the limiting
Proportional-cost function $\psi_{\Prop}\colon(0,\infty)\to\reals$ by
\begin{align*}
  \psi_{\Prop}(t)
  \deq\frac{6t^{2}+10t+2}{2t^{2}+3t+1}.
\end{align*}
The expected mechanism costs on the spread-out profile then satisfy
\begin{align}\label{eq:family-spread}
  \SC(\GPair,\prof_{\mathrm{s}}^{A,B})&\longrightarrow4,
  &
  \SC(\Prop,\prof_{\mathrm{s}}^{A,B})&\longrightarrow\psi_{\Prop}(t).
\end{align}
The derivative of the limiting Proportional-cost function is
\begin{align*}
  \psi_{\Prop}'(t)
  =\frac{-2t^{2}+4t+4}{(2t^{2}+3t+1)^{2}}.
\end{align*}
Hence the limiting Proportional-cost function is uniquely maximized at the
mass ratio $t_{*}\deq1+\sqrt3$, where
\begin{align*}
  \psi_{\Prop}(t_{*})=10-4\sqrt3.
\end{align*}
Along any sequence of spread-out profiles whose mass ratio tends to $t_{*}$,
the normalized mixture cost therefore converges to
\begin{align*}
  \frac{\SC(\Mix{\lambda},\prof_{\mathrm{s}}^{A,B})}
       {\OPT(\prof_{\mathrm{s}}^{A,B})}
  &\longrightarrow
    \lambda\cdot\psi_{\Prop}(t_{*})+4\cdot(1-\lambda) \\
  &=\lambda\cdot(10-4\sqrt3)+4\cdot(1-\lambda)
   =4-\kcon\cdot\lambda.
\end{align*}
Taking the supremum over profiles proves the spread-out-profile lower bound
\begin{align}\label{eq:frontier-spread}
  \ratio(\Mix{\lambda})\ge4-\kcon\cdot\lambda.
\end{align}

Combining the concentrated-profile bound \cref{eq:frontier-concentrated} and
the spread-out-profile bound \cref{eq:frontier-spread} gives
$\ratio(\Mix{\lambda})\ge L(\lambda)$.  Finally, the concentrated-profile
branch $3+\lambda$ is increasing in the mixing weight $\lambda$, whereas the
spread-out-profile branch $4-\kcon\cdot\lambda$ is decreasing in the mixing
weight $\lambda$.  The two branches meet at the balancing weight
$\lamstar\in(0,1)$, and their common value is the balanced frontier value
$\rhostar$.  Thus the lower-frontier function $L(\lambda)$ is minimized at the
balancing weight $\lamstar$ with value $\rhostar$.  Therefore, the worst-case
approximation ratio $\ratio(\Mix{\lambda})$ is at least $\rhostar$ for every
mixing weight $\lambda\in[0,1]$; taking the infimum over the mixing weight
proves the claimed consequence.
\end{proof}

The two profile families expose the complementary failure modes directly.  The
concentrated-profile family constrains large mixing weights $\lambda$ through
the concentrated-profile branch $3+\lambda$, whereas the spread-out-profile
family constrains small mixing weights $\lambda$ through the spread-out-profile
branch $4-\kcon\cdot\lambda$.  Combining the exact $11/3$ upper bound in
\Cref{thm:mixture-exact} with the lower-frontier bound in \Cref{prop:frontier},
we therefore know that the best ratio within the mixture family lies between
the balanced frontier value $\rhostar=3.5186\ldots$ and $11/3=3.6667\ldots$.
The exact optimizing weight and the exact value within this interval remain
open.
 
\end{document}